%% file: dimmatch-arxiv.tex
\documentclass[10pt]{article}

\pdfoutput=1

\usepackage{textcase}

\usepackage[lmargin=1in,rmargin=2in,tmargin=0.75in,bmargin=0.75in,letterpaper]{geometry}
\usepackage[tiny,compact,raggedright]{titlesec}
\titleformat*{\subsection}{\itshape}
\titleformat*{\section}{\bfseries}
\let\oldsection=\section
\renewcommand{\section}[1]{\ifthenelse{\equal{#1}{*}}%
  {\oldsection*}%
  {\oldsection{\texorpdfstring{\MakeTextUppercase{#1}}{#1}}}%
}

\usepackage{url}

\usepackage{amsmath}
\usepackage{amssymb}

\usepackage{graphicx}

\usepackage{cite}

\usepackage[svgnames]{xcolor}

\usepackage[labelfont=bf,labelsep=period,justification=raggedright]{caption}

\usepackage{booktabs}
\usepackage{pdfpages}
\usepackage{tabularx}
\usepackage{paralist}
\usepackage{xcolor}
\usepackage{graphicx}
\graphicspath{{./}{./figures/}}
\usepackage{epstopdf}
\usepackage{amsmath}
\usepackage{nicefrac}

\usepackage{amsthm}
\theoremstyle{plain}
\newtheorem{theorem}{Theorem}

\newtheorem{lemma}{Lemma}
\newtheorem*{note}{Note}

\usepackage[colorlinks,linkcolor=DarkBlue,citecolor=DarkBlue]{hyperref}

\newcommand{\set}[1]{\ensuremath{\left\{ #1 \right\}}}
\newcommand{\norm}[2][]{\ensuremath{\left\| #2 \right\|_{#1}}}
\newcommand{\size}[1]{\ensuremath{\left| #1 \right|}}
\newcommand{\lilOh}[1]{\ensuremath{o\!\left( #1 \right)}}
\newcommand{\bigOh}[1]{\ensuremath{\mathcal{O}\!\left( #1 \right)}}
\newcommand{\bigOmega}[1]{\ensuremath{\Omega\!\paren{#1}}}
\newcommand{\prob}[1]{\ensuremath{\mathbb{P}\!\left( #1 \right)}}
\newcommand{\expect}[1]{\ensuremath{\mathbb{E}\!\left[ #1 \right]}}
\newcommand{\abs}[1]{\ensuremath{\left| #1 \right|}}
\newcommand{\paren}[1]{\ensuremath{\left( #1 \right)}}
\DeclareMathOperator{\MGEOP}{MGEO-P}
\newcommand{\mgeop}{\ensuremath{\MGEOP(n,m,\alpha,\beta,p)}}
\newcommand{\bigTheta}[1]{\ensuremath{\Theta\!\paren{#1}}}

\newcommand{\N}{\ensuremath{\mathbb{N}}}

\newcommand{\TbigOh}[1]{\ensuremath{\tilde{\mathcal{O}}\!\paren{#1}}}
\newcommand{\TbigTheta}[1]{\ensuremath{\tilde{\Theta}\!\paren{#1}}}
\newcommand{\floor}[1]{\ensuremath{\left\lfloor #1 \right\rfloor}}

 \title{Dimensionality of social networks\\using motifs and eigenvalues}

 \author{Anthony Bonato \and David F.~Gleich \and Myunghwan Kim \and Dieter Mitsche,$^{4}$ \and Pawe\l\ Pra{\l}at \and Amanda Tian \and Stephen J.~Young}

\begin{document}

\vspace*{\baselineskip}

\noindent \begin{minipage}{0.7\linewidth}
 \raggedright
 \Large \textsc{dimensionality of social networks using motifs and eigenvalues}
\end{minipage} \bigskip

\noindent \begin{minipage}{0.7\linewidth}
\itshape\footnotesize
 Anthony Bonato,$^{1\ast}$
David F.~Gleich,$^{2\ast}$
Myunghwan Kim,$^{3}$\\
Dieter Mitsche,$^{4}$
Pawe\l\ Pra{\l}at,$^{1}$
Amanda Tian,$^{1}$
Stephen J.~Young$^{5}$
\end{minipage} \bigskip

\footnotetext[1]{Department of Mathematics, Ryerson University, Toronto, ON Canada }
\footnotetext[2]{Computer Science Department, Purdue University, West Lafayette, IN USA }
\footnotetext[3]{Electrical Engineering Department, Stanford University, Stanford, CA USA }
\footnotetext[4]{Laboratoire J.A.~Dieudonn\'{e}, Universit\'{e} de Nice Sophia-Antipolis, Nice, France}
\footnotetext[5]{Mathematics Department, University of Louisville, Louisville, KY USA}
\renewcommand*{\thefootnote}{\fnsymbol{footnote}}
\footnotetext[1]{Corresponding authors, \url{abonato@ryerson.ca}, \url{dgleich@purdue.edu} }

\vspace*{\baselineskip}


\section*{Abstract}

We consider the dimensionality of social networks, and develop experiments aimed at predicting that dimension. We find that a social network model with nodes and links sampled from an \textit{m}-dimensional metric space with  power-law distributed influence regions best fits samples from real-world networks when \textit{m} scales logarithmically with the number of nodes of the network. This supports a logarithmic dimension hypothesis, and we provide evidence with two different social networks, Facebook and LinkedIn. Further, we employ two different methods for confirming the hypothesis: the first uses the distribution of motif counts, and the second exploits the eigenvalue distribution.

\bigskip

\section{Introduction}

Empirical studies of on-line social networks as undirected graphs suggest these graphs have several intrinsic properties: highly skewed or even power-law degree distributions~\cite{barabasi1999-scaling,Faloutsos-1999-power-law}, large local clustering~\cite{watts1998-dynamics}, constant~\cite{watts1998-dynamics} or even shrinking diameter with network size~\cite{Leskovec-2007-densification}, densification~\cite{Leskovec-2007-densification}, and localized information flow bottlenecks~\cite{Estrada-2006-expansion,Leskovec-2009-community-structure}. Many existing  models of social network connections and growth have trouble capturing all of these properties simultaneously~\cite{Kim-2012-mag,Kolda-2013-BTER,Gleich-2012-kronecker}. One that does is the geometric protean model (GEO-P)~\cite{Bonato-2012-geop}. It differs from other network models~\cite{Kumar-2000-copying,barabasi1999-scaling,Leskovec-2007-densification,Leskovec-2010-KronFit} because all links in geometric protean networks arise based on an underlying metric space. This metric space mirrors a construction in the social sciences called \emph{Blau space}~\cite{McPherson-1991-Blau}. In Blau space, agents in the social network correspond to points in a metric space, and the relative position of nodes follows the principle of \emph{homophily}~\cite{McPherson-2001-homophily}: nodes with similar socio-demographics are closer together in the space.

In order to accurately capture the observed properties of social networks---in particular, constant or shrinking diameters---the dimension of the underlying metric space in the GEO-P model must grow logarithmically with the number of nodes. The logarithmically scaled dimension is a property that occurs frequently with network models that incorporate geometry, such as in multiplicative attribute graphs~\cite{Kim-2012-mag} and random Apollonian networks~\cite{Zhang-2006-Apollonian}. Because of its prevalence in these models, the logarithmic relationship between the dimension of the metric space and the number of nodes has been called the \emph{logarithmic dimension hypothesis}~\cite{Bonato-2012-geop}. This hypothesis generalizes previous analysis which shows that individuals in a social network can be identified with relatively little information. For instance, Sweeney found that 87\% of the U.S.~population had reported attributes that likely made them unique using only zip code, gender and date of birth, and concluded that few attributes were needed to uniquely identity a person in the U.S. population~\cite{Sweeney-2000-uniqueness}. In the following study, we find evidence of the log-dimension property in real world social networks.

We emphasize that the present paper is the first study that we are aware of which attempts to quantify the dimensionality of social networks and Blau space. While we do not claim to prove conclusively the logarithmic dimension hypothesis for such networks, our experiments, such as those of \cite{Sweeney-2000-uniqueness}, suggest a much smaller dimension in contrast to the overall size of the networks. Interestingly, speculation on the low dimensionality of social networks arose independently from theoretical analysis of mathematical models of social networks in~\cite{Bonato-2012-geop,Kim-2012-mag,Zhang-2006-Apollonian}.

\subsection{MGEO-P}
The particular network model we study is a simple variation on the GEO-P model that we name the memoryless geometric protean model (MGEO-P), since it enables us to approximate a GEO-P network without using a costly sampling procedure. The MGEO-P model depends on five parameters:
\begin{center}
\begin{tabular}{cl}
$n$ & the total number of nodes,\\
$m$ & the dimension of the metric space,\\
$0 < \alpha < 1$ & the attachment strength parameter, \\
 $0 < \beta < 1-\alpha$ & the density parameter, \\
 $0 < p \le 1$ & the the connection probability.
\end{tabular}
\end{center}
The nodes and edges of the network arise from the following process. Initially the network is empty. At each of $n$ steps, a new node $v$ arrives and is assigned both a random position $p_v$ in $R^m$ within the unit-hypercube $[0,1]^m$ and a random rank $r_v$ from those unused ranks remaining in the set $1$ to $n$. The influence radius of any node is computed based on the formula:
\[ I(r) = \tfrac{1}{2} \bigl(r^{-\alpha} n^{-\beta}\bigr)^{1/m}. \]
With probability $p$, the node $v$ forms an undirected connection to any preexisting node $u$ where $D(v,u) \le I(r_v)$, where the distances are computed with respect to the following metric:
\[ D(v,u) = \min\set{\norm[\infty]{p_v-p_u -z} \colon z \in \set{-1,0,1}^m },  \]
and where $\norm[\infty]{\cdot}$ is the infinity-norm.
We
note that this implies that the geometric space is symmetric in
any point as the metric ``wraps'' around like on a torus. The volume of space influenced by the node is $r_v^{-\alpha} n^{-\beta}$. Then the next node arrives and repeats the process until all $n$ nodes have been placed.

Figure~\ref{fig:geop-fig} illustrates two features of the model. First, after a few steps, only a few nodes exist and even a large influence region will only produce a few links. Second, when the number of steps approaches $n$, a large influence region will produce many links. The idea behind the model is a simple abstraction of the growth of an on-line social network. When the network is first growing (few steps), even influential members will only know a few other members who have also joined. But after the network has been around for a while (many steps), influential members will begin with many friends.

\begin{figure}
\centering
 \includegraphics{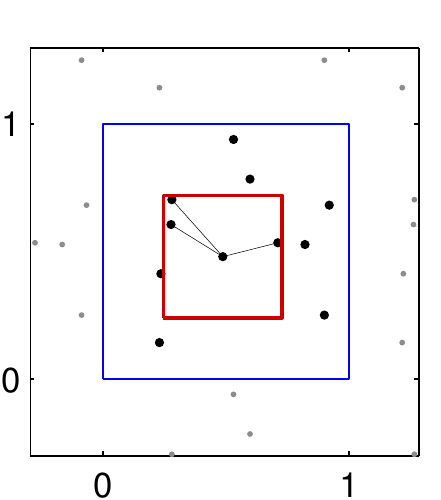}%
 \includegraphics{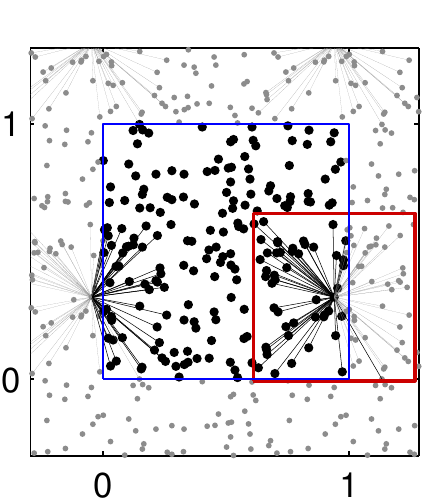}%
 \includegraphics{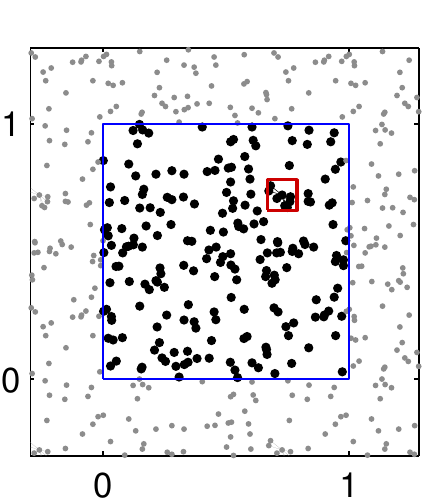}

 \includegraphics[width=3in]{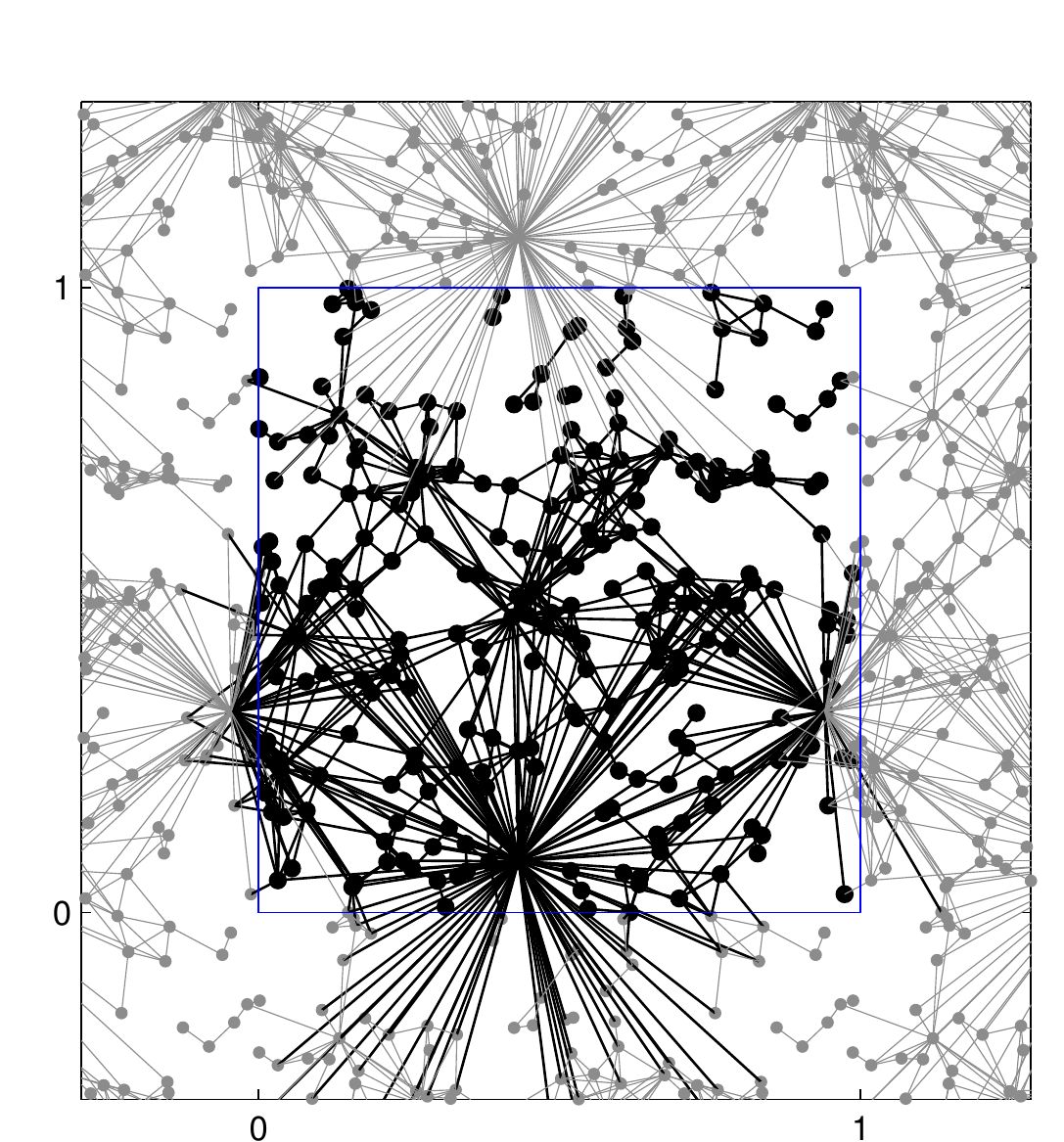}

 \caption{An example describing the MGEO-P process on a graph with $250$ nodes in the unit square with torus metric, where $\alpha=0.9$ and $\beta = 0.04$ and $p=1$. Each figure shows the graph ``replicated'' in grey on all sides in order to illustrate the torus metric. Links are drawn to the closest replicated neighbor. The blue square indicates the region $[0,1]^2$. \emph{Top row (left to right)} The MGEO-P process begins with relatively few nodes, and thus, nodes must have large influence radii (red squares) to link anywhere. As more nodes arrive, large radii result in many connections, modeling influential users, and small radii result in a few connections, modeling standard users. \emph{Bottom row} Illustrates the final constructed graph.}
 \label{fig:geop-fig}
\end{figure}

We formally prove that the MGEO-P model has the following properties. Let $\alpha \in (0,1), \beta \in (0,1-\alpha), p \in (0,1]$ and $m$ be positive integer. The following statements hold with probability tending to $1$ as $n$ tends to $\infty$:\footnote{See the MGEO-P section of the appendix for the proofs. We actually show these results hold with extremely high probability, which is a stronger notion that implies probability tending to $1$.}
\begin{enumerate}
\item Let $v$ be a node of $\mgeop$ with rank $R$ that arrived at step $t$. Then
\[ \begin{aligned} \deg(v) = & \paren{\frac{i-1}{n-1}\frac{p}{1-\alpha}n^{1-\alpha-\beta} +
(n-i)pR^{-\alpha}n^{-\beta}}\\
& \quad \cdot \paren{1+\bigOh{\sqrt{\frac{\log^2(n)}{n^{1-\alpha-\beta}}}}}
\end{aligned}. \] This result implies that the degree distribution follows a powerlaw with exponent $\eta = 1 + \frac{1}{\alpha}$.
\item The average degree of node of $\mgeop$ is
\[\rho = \frac{p}{1-\alpha}n^{1-\alpha-\beta}\paren{1+\bigOh{\sqrt{\frac{\log^2(n)}{n^{1-\alpha-\beta}}}}}.\]
\item The diameter of $\mgeop$ is $n^{\bigTheta{\frac{1}{m}}}$.
\end{enumerate}
This last property suggests that, ignoring constants,
for a network with $n$ nodes and diameter $D$, the expected dimension based on the MGEO-P model is
\[ m \approx \frac{\log n}{\log D}. \]
Thus, like some network models that incorporate geometry~\cite{Kim-2012-mag,Zhang-2006-Apollonian}, in the MGEO-P model, the dimension $m$ must scale logarithmically in order for the diameter to remain constant as $n$ increases.


\subsection{Experimental Design and Graph Summaries}

Both graph motifs and spectral densities are numeric summaries of a graph that abstract the details of a network into a small set of values that are independent of the particular nodes of a network. These summaries have the property that isomorphic graphs have the same values, and 
we will use these summaries to determine the dimension of the metric space that best matches Facebook and LinkedIn networks as illustrated in Figure~\ref{fig:pipeline}. Graph motifs, graphlets, or graph moments are the frequency or abundance of specific small subgraphs in a large network. We study undirected, connected subgraphs up to four nodes as our graph motifs. This is a set of 8 graphs shown in at the bottom of Figure~\ref{fig:pipeline} along with the single two node graph of an edge. The spectral density of a graph is the statistical distribution of eigenvalues of the normalized Laplacian matrix as indicated in the upper right of that figure. These eigenvalues indicate and summarize many network properties including the behavior of a uniform random walk, the number of connected components, an approximate connectivity measure, and many other features~\cite{Chung-1992-book,Banerjee-2009-graph-spectra}. Thus, the spectral density of the normalized Laplacian is a particularly helpful characterization that captures many such separate network properties.

\begin{figure}
\centering
\includegraphics[width=0.8\linewidth]{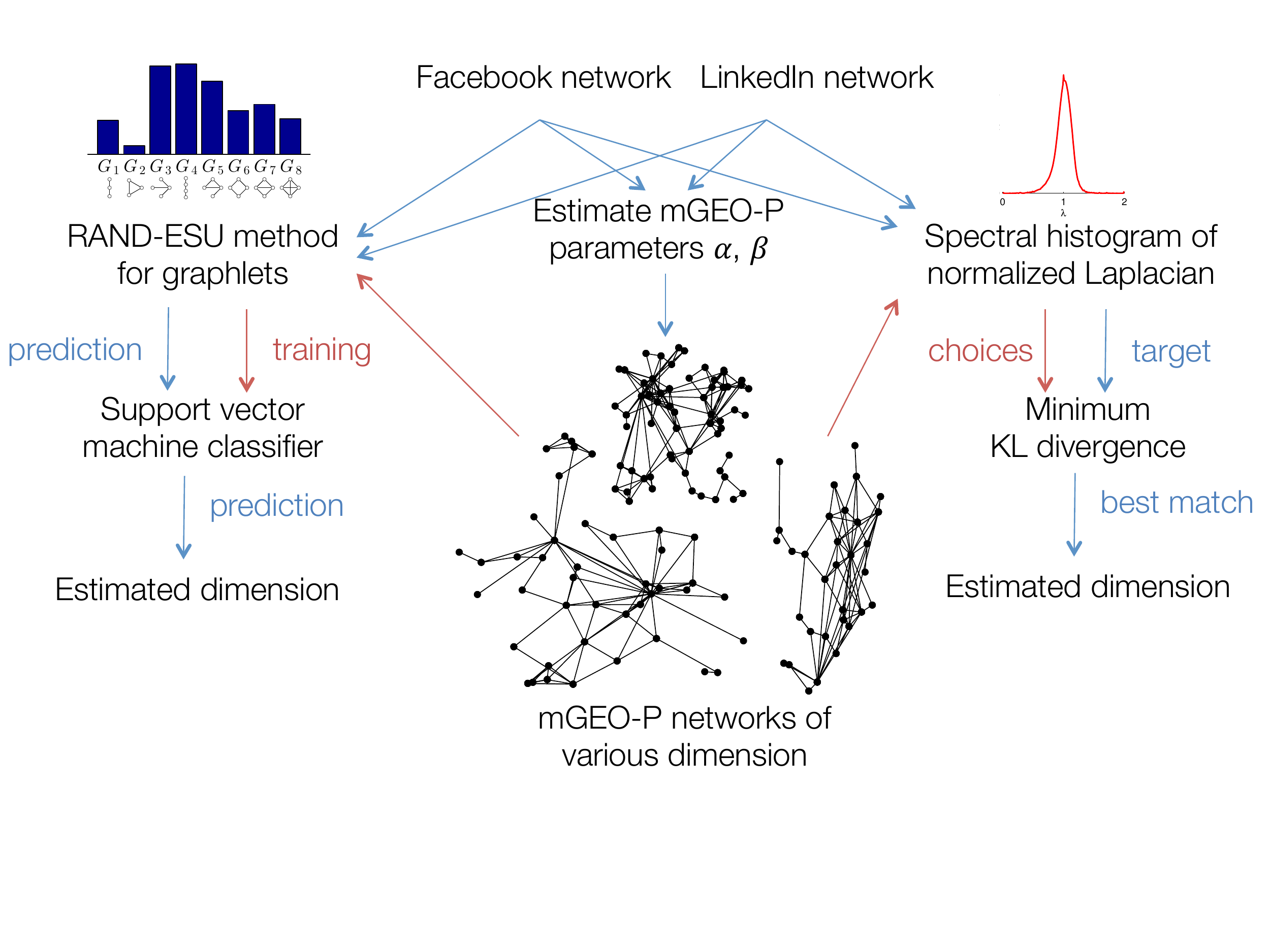}

\vspace{-1.5cm}
\includegraphics[width=0.8\linewidth]{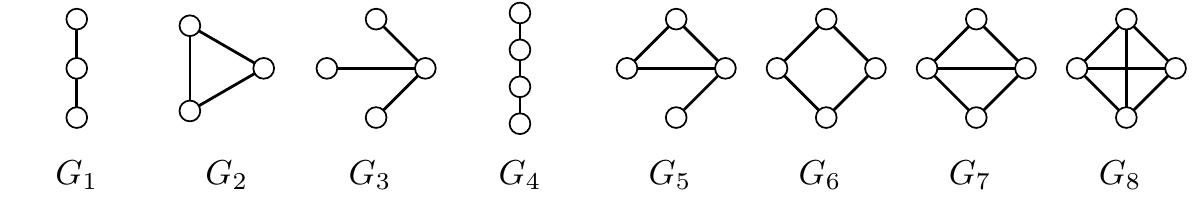}

\caption{At left and center, we have the steps involved in fitting via graphlets; at right and center, we have the steps involved in fitting via spectral histogram. Throughout, red lines denote the flow of features for the MGEO-P networks whereas blue lines denote flow of features for the original networks. At the bottom, we show an enlarged representation of the 8 graphlets we use.}
\label{fig:pipeline}
\end{figure}

We study dimensional scaling in social networks by comparing samples of the \mbox{MGEO-P} networks of varying dimensions with samples of social network data from Facebook and LinkedIn. We pay particular attention to the relationship between the number of nodes $n$ of the network and the dimension $m$ of the best fit MGEO-P network. In order to determine what underlying dimension for MGEO-P best fits a given graph, we employ two distinct methods. For one experiment, we use features known as graph motifs, graphlets, or graph moments in concert with a support vector machine (SVM) classifier. This approach has been used successfully to determine the best generative mechanism of a network~\cite{Memisevic-2010-modeling} and to select parameters of a complicated network models to fit real-world data~\cite{Gleich-2012-kronecker,Moreno-2013-mom}.  In a second experiment, we use spectral densities of the normalized Laplacian matrix of a graph and a KL-divergence similarity measurement, which has been used to match protein networks between species~\cite{Patro-2012-ghost,Banerjee-2012-distance}. We find evidence of the logarithmic dimension hypothesis in both cases.

\subsection{The data}
Facebook distributed 100 samples of social networks from universities within the United States measured as of September 2005~\cite{Traud-2011-facebook}, which range in size from 700 nodes to 42,000 nodes. We call these networks the Facebook samples. The LinkedIn samples were created from the LinkedIn  connection network together with the creation time of each connection from May 2003 to October 2006. To perform our experiments on networks of different size, we build the snapshots of the LinkedIn network at various timestamps. We then extracted a dense subset of their graph at various time points that is representative of active users; we used the 5-core of the network for this purpose~\cite{Seidman1983-cores}. See Figure~\ref{fig:data-scaling} and the appendix for additional properties of these networks. In both networks, the number of edges per node grows at essentially the same rate.

\begin{figure}[t]
\centering
\includegraphics{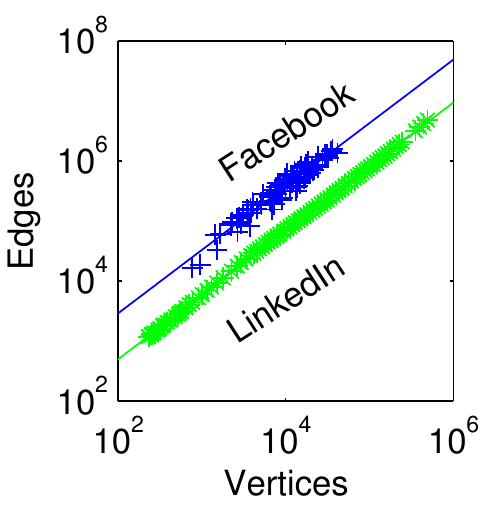}
\caption{The scale of the network data involved in our study varies over three orders of magnitude. We see similar scaling for both types of networks, but with slightly different offsets. For Facebook, $\log_{10}(\text{edges}) = 1.06 \log_{10}(\text{nodes}) + 1.35$ with $R^2 = 0.945$; for LinkedIn $\log_{10}(\text{edges}) = 1.07 \log_{10}(\text{nodes}) + 0.56$ with $R^2 > 0.999$. The regularity in the LinkedIn sizes is due to our construction of those networks.}
\label{fig:data-scaling}
\end{figure}

\begin{figure}[t]
\centering
\includegraphics[width=0.4\linewidth]{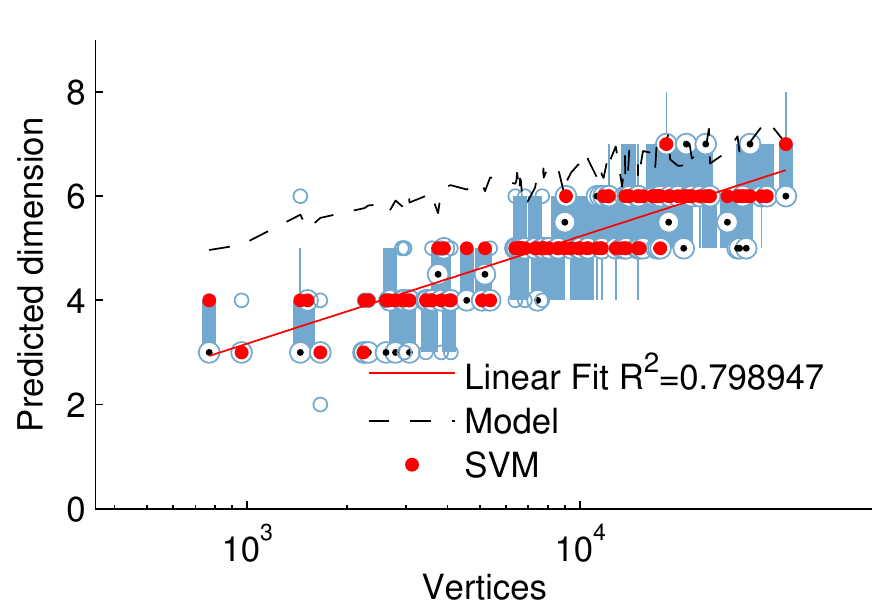}
\includegraphics[width=0.4\linewidth]{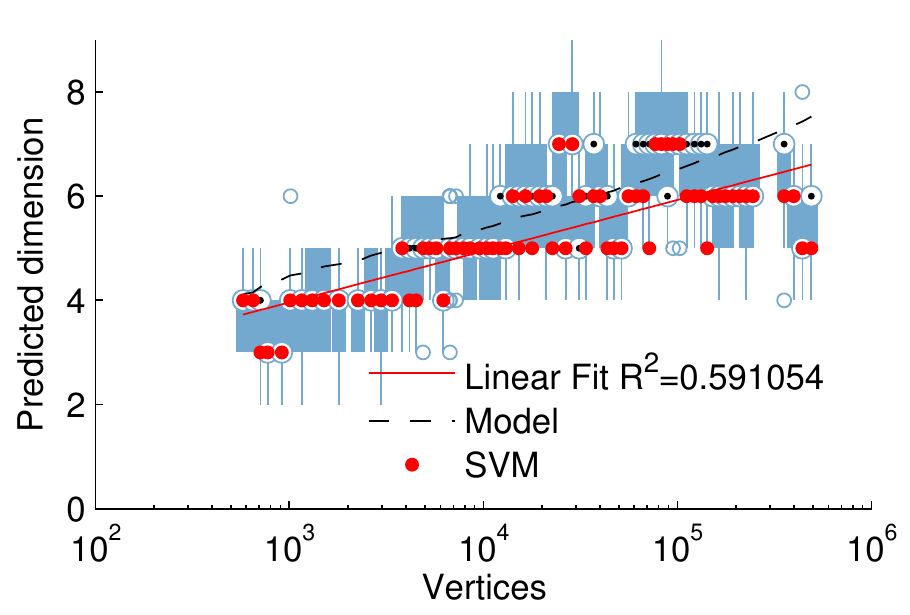}
\caption{Facebook dimension at top, LinkedIn dimension at bottom---computed via graphlet features and a support vector machine classifier to select the dimension. For the Facebook data, we find that $m = 2.06 \log(n)/ \log(10) - 3.00$.
For the LinkedIn data, we find that $m = 0.7333 \log(n)/ \log(10) + 1$. In the left figure, we show the variance in the fitted dimension as a box-plot.  We estimate the variance by using only 20\% of the original training data and repeating over 50 trials. There are only a few outliers for small dimensions. }
\label{fig:facebook-dim}
\end{figure}

\begin{figure}[t]
\centering
\includegraphics[width=0.4\linewidth]{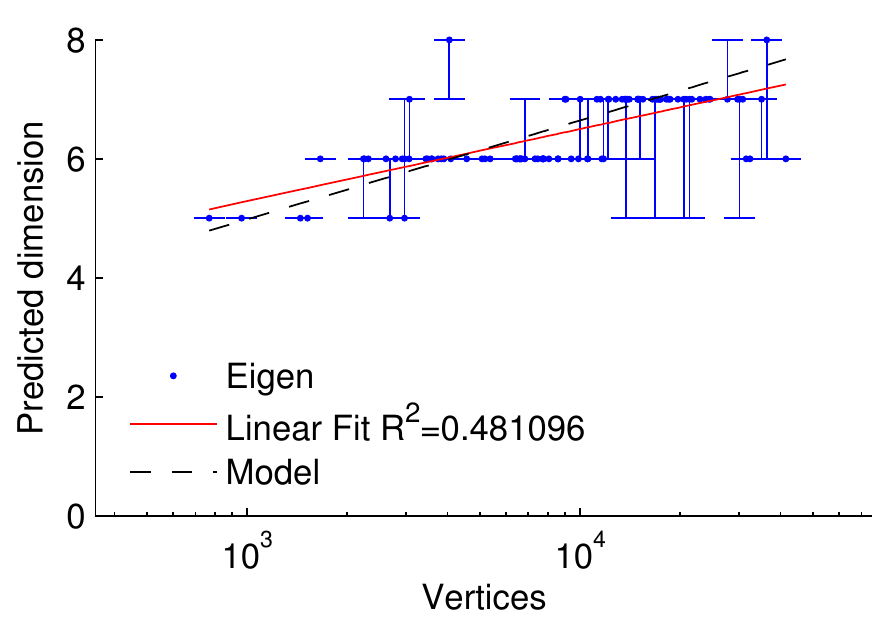}
\includegraphics[width=0.4\linewidth]{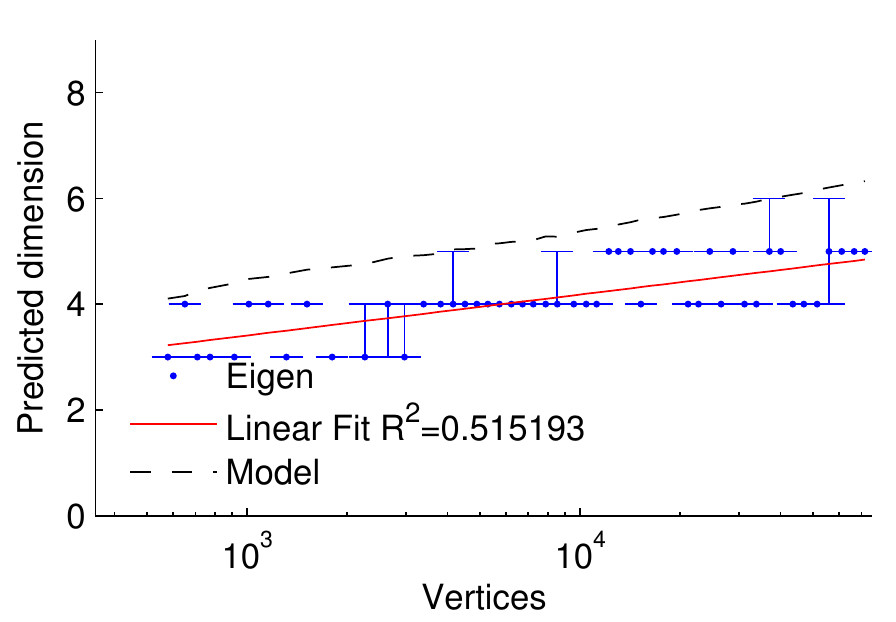}
\caption{At top, Facebook data, at bottom, LinkedIn data. We show the fitted dimensions based on the minimum KL-divergence between the spectral densities. The dimensions shift modestly higher for Facebook and remain almost unchanged for LinkedIn. Both still are closely correlated with the theoretical prediction based on the model.}
\label{fig:dims-eigs}
\end{figure}

\section{Results}

The results of our dimensional fitting for graphlets are shown in Figure~\ref{fig:facebook-dim} and the results of the fitting using spectral densities are in Figure~\ref{fig:dims-eigs}. For both datasets and both types of statistics, the best-fit dimension scales logarithmically with the number of nodes and closely tracks a simple model prediction based on the diameter $D$ of the network (the model curve plots $m = \log(n)/\log(D)$). These experiments corroborate the logarithmic dimension hypothesis; although the precise fits differ:

\emph{Using graphlets}, for the Facebook data, we find that the dimension $m = 2.06 \log(n)/\log(10) - 3.00 $ with 95\% confidence intervals of $ (1.851, 2.264)$ and $(-3.821, -2.182)$, respectively.
For the LinkedIn data, we find that $m = 0.98 \log(n)/\log(10) + 1.01 $ with 95\% confidence intervals  of $(0.786, 1.178)$ and $(0.1591, 1.87)$. \emph{Using spectral densities}, for the Facebook networks, we find that $d = 1.21 \log(n) / \log(10) + 1.65$ is the best-fit line, with a 95\% confidence interval for the coefficients of $(0.9782, 1.446)$ and $(0.7242, 2.578)$. For the LinkedIn networks, we find $d = 0.77 \log(n) / \log(10) + 1.1$. The 95\% confidence interval for these coefficients, respectively is  $(0.56, 0.99)$ and  $(0.23, 1.95)$.

\subsection{Sensitivity}
We investigate the sensitivity of the graphlet results in two settings. If we reduce the training set size of the SVM classifier by using a random subset of 20\% of the input training data and then rerun the training and classification procedure 50 times, then we find a distribution over dimensions that we report as a box-plot, shown in Figure~\ref{fig:facebook-dim}. In the appendix, we further study perturbation results that argue against these results occurring due to chance. In particular, we find that these dimensions are robust to moderate changes to the network structure and we find that our methodology does not predict useful dimensions of
$\mathrm{\mbox{Erd}\ddot{o}\text{s-R}\acute{e}\text{nyi}}$
random graphs or random graphs with the same degree distribution.  We do not report a precise $p$-value as there are no widely accepted null-models for network data. We study the sensitity of the spectral densities that look for matches that are within $105\%$ of the true minimum divergence. This defines a dimension interval around each match that is small for all of our examples.

\section*{Discussion}

There is a growing body of evidence that argues for some type of geometric structure in social and information networks. An important study in this direction views networks as samples of geometric graphs within a hyperbolic space~\cite{Krioukov-2010-hyperbolic,Krioukov-2012-cosmology,Krioukov-2013-growing}. Recent work has further shown that hyperbolic embeddings reproduce  shortest path metrics in real-world networks~\cite{Zhao-2011-rigel}. In both MGEO-P and hyperbolic random geometric networks, highly skewed or power-law degree distributions are imposed---either directly as in MGEO-P, or implicitly as in the hyperbolic space scaling. These results further support hidden metric structures in networks by empirically confirming a prediction about the dimension of the metric space made by one particular model.

Note that these results do not conclusively argue that MGEO-P is a \textbf{perfectly accurate} model for social networks; there are meaningful differences between the spectral histograms from MGEO-P and real social networks, see Figure~\ref{fig:histograms}. There are also similar differences in the graphlet counts. Our results support a \textbf{different} hypothesis. The closest MGEO-P network to a given social network has a metric space whose dimension scales logarithmically with the number of nodes. In the appendix material we have determined that this property is not due to either the edge density or the degree distribution; thus, our findings appears to reflect a new intrinsic property of social networks.

\begin{figure}[t]
\centering
\includegraphics[width=0.25\linewidth]{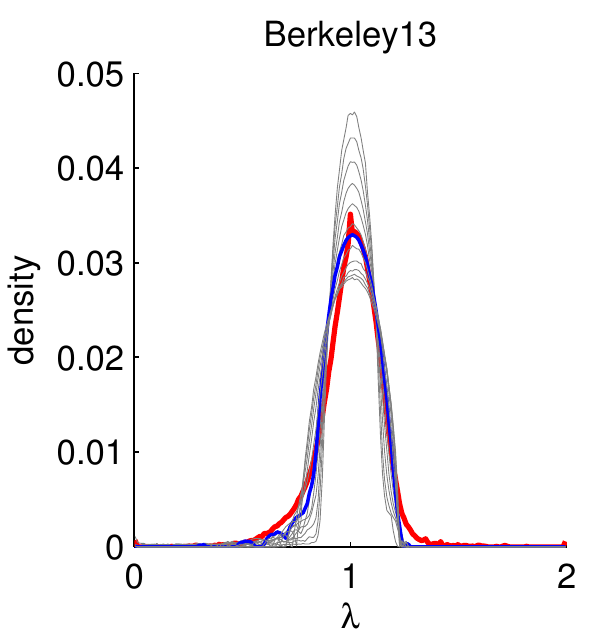}%
\includegraphics[width=0.25\linewidth]{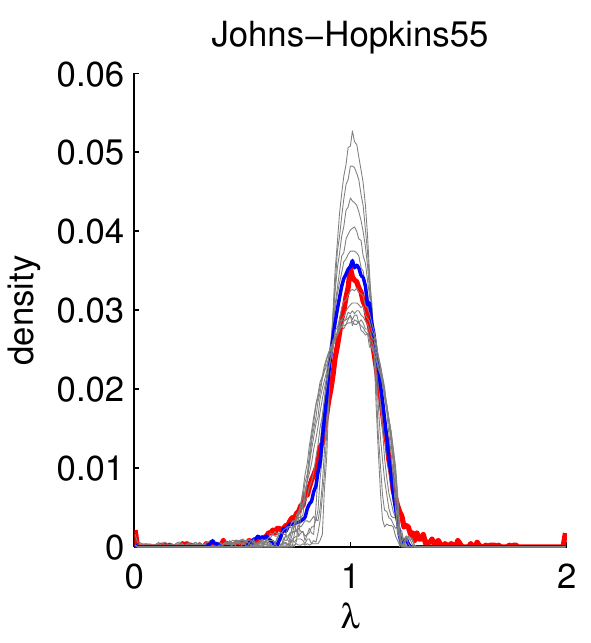}%
\includegraphics[width=0.25\linewidth]{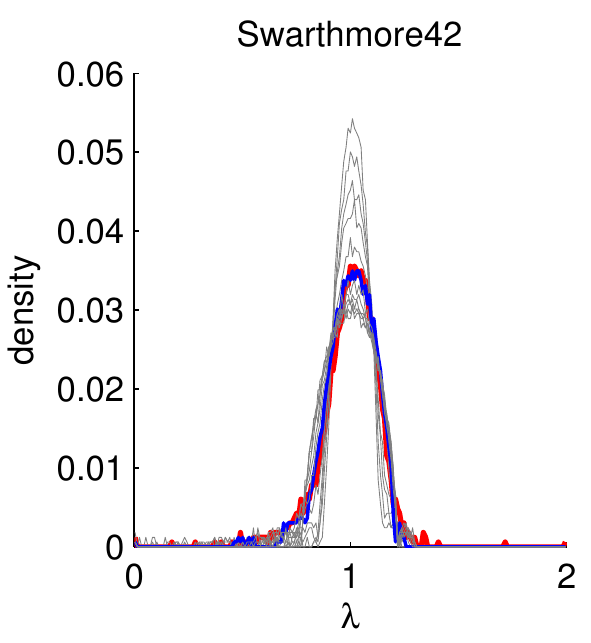}
\caption{For three of the Facebook networks, we show the eigenvalue histogram in red, the eigenvalue histogram from the best fit MGEO-P network in blue, and the eigenvalue histograms for samples from the other dimensions in grey. The MGEO-P model correctly captures the peak of the distribution around 1, but fails to completely capture the tail between 1 and 2. Thus, we see meaningful difference between these profiles and hence, do not suggest that MGEO-P captures all of the properties of real-world social networks.}
\label{fig:histograms}
\end{figure}

\section{Experimental design}

Given a graph $G = (V,E)$, we employ the following methods to determine the dimension $m$ of the MGEO-P models:
\begin{quote} \textbf{Experiment 1} \\
1. Set $n$ to the number of nodes. Determine values of $\alpha$ and $\beta$ independently of $m$ (see the appendix of the original paper). \\
2. Simulate $50$ samples of an MGEO-P network with $m$ varying between $1$ and $12$.  \\
3.  Compute the graphlet counts for each sample of MGEO-P and train a SVM classifier to predict the dimension of the network given the samples. \\
4. Compute the graphlet counts for the graph $G$ and use the output from the classifier as the dimension $m$ of the network.
\end{quote}

\smallskip

\begin{quote} \textbf{Experiment 2} \\
1 \& 2. As in experiment 1. \\
3. Compute the spectral density for one sample of MGEO-P for each $m$ between $1$ and $12$ (only one MGEO-P sample is used to get the density).\footnote{We only use one sample of the MGEO-P network to estimate the eigenvalue distribution as these computations are time-consuming and our preliminary studies showed that the spectral density had only small variations between repeated samples.} \\
4. Compute the spectral density of the graph $G$ and find the value of $m$ that minimizes the KL-divergence between the density from the graph and the MGEO-P samples.
\end{quote}
The first approach employs a complex statistical technique---the support vector machine classifier---to determine nonlinear predictive correlations among the graphlet counts and the dimension. This sophistication renders the method opaque and difficult to interpret the precise similarity mechanism. The second approach is simple and still illustrates the dimensional scaling, although the precise dimensions differ, which indicates that it is matching the network in a different way.

\subsection{Estimating dimensions using graphlets and support vector machines}

The relationship between the dimension of a graph and its graphlets is highly nonlinear and so we used a multi-class support-vector machine (SVM) based classification tool from WEKA to predict this relationship.
In this case, each dimension is a class, but as an SVM can only make a binary decision we train the SVM using a dimension-vs-dimension classification.
That is, we build a classifer to predict dimension 5-vs-dimension 3, dimension 5-vs-dimension 4, etc. so there are 66 = ``12-choose-2'' SVMs trained. The dimension picked most often among these classifiers is the predicted class; this is the standard behavior of the sequential minimal optimization classifier (SMO) used in Weka. The dimension of a real-world network is then predicted by running this classifier on the graphlet counts of the networks.  An alternative methodology (which has had some previous success) would be to to train the classifier using alternating decision trees; however this training methodology significantly restricts the behavior of the classifier and produces inconsistent results.

\subsection{Comparing spectral densities}

Given the eigenvalues of the normalized Laplacian, we compute a spectral density by taking a 201-bin histogram of these eigenvalues. We then use the KL-divergence between these histograms as used in Banerjee and Jost (2009) as a measure of similarity. If $P^A$ and $P^B$ are the histograms of networks $A$ and $B$ normalized to probabilities, then for our $201$-bin histograms we have that:
\[ KL(A,B) = \sum_{i=1}^{201} \log(P^A_i / P^B_i) P^A_i. \]
We select the single best dimension based on the value of $m$ that minimizes the KL divergence $KL(S,G_m)$ where $S$ is the sample of either Facebook or LinkedIn and $G_m$ is a sample of a MGEO-P network with dimension $m$. We add $1$ to all of the eigenvalue counts in the histogram as a form of smoothing for the probabilities. We define a dimension interval by looking at the maximum interval such that the extreme points are within $105\%$ of the true minimum.

\subsection{Specific methods}

\paragraph{Powerlaw fitting} To determine the powerlaw exponent $\eta$, we use the Clauset-Shalizi-Newman power-law exponent estimator~\cite{clauset2009-powerlaw} as implemented by Tam\'{a}s Nepusz~\cite{Nepusz-plfit}.

\paragraph{Diameters} The MGEO-P model of a network predicts that the dimension $m$ should approximate $\log(n) / \log(D)$, where $D$ is the diamater. However, as $D$ is sensitive to outliers we use the 99\% effective diameter computed via an asymptotically accurate approximation scheme~\cite{Palmer-2002-fast-anf} as implemented in the SNAP library on 2011-12-31. The effective diameter of all Facebook networks ranges between 3.5 and 4.6, with a mean of 4.1.  For the LinkedIn data, the effective diameter ranges between 4.3 and 5.9, with a mean of 5.4. In both networks, larger graphs have bigger effective diameters, although the differences are slight and the full data is available in the appendix material.

\paragraph{Graphlets} To compute graphlets, we employ the rand-esu sampling algorithm~\cite{Wernicke-2006-motifs} as implemented in the igraph library~\cite{Csardi-2006-igraph}. This algorithm approximates the count of each subgraph via a stochastic search, which then depends on the probability of continuing to search. Thus, if the probability is near $1$ then the scores are nearly exact, but very expensive to compute, and small probabilities truncate the search early to produces fast estimates. The value we use is $10/n$. We use log-transformed output from this procedure in order to capture the dynamic range of the resulting values.

\paragraph{Spectral densities} We approximate the spectral density via a 201-bin histogram of the eigenvalues of the normalized Laplacian, which all fall between $0$ and $2$. (The choice of 201 was based on prior experiences with the spectral histograms of networks.)  To compute eigenvalues of a network, we employ the recently developed ScaLAPACK routine using the MRRR algorithm~\cite{Dhillon-1997-mrrr,Dhillon-2006-mrrr,Vomel-2010-MRRR}.

\paragraph{SVM} We used a multi-class support-vector machine (SVM) based classification tool from Weka~\cite{witten2005-weka} to predict the relationship between the graphlets and the dimension.

\paragraph{Setting MGEO-P Parameters}
Consider a graph $G = (V,E)$ that we wish to compare to an MGEO-P sample. The MGEO-P model depends on four parameters: $n$, $m$, $\alpha$, and $\beta$. The choice of $n$ is straightforward as we use the number of nodes of the original graph. Both $\alpha$ and $\beta$ can be chosen independently of the dimension $m$. Specifically, both $\alpha$ and $\beta$ determine the average degree of the network and the exponent of the power law in the degree distribution, up to lower-order terms, as shown by property 1 and property 2. By computing just these two simple statistics of a network---the exponent of the power law and the average degree---we can invert these relationships and choose these parameters.  Let $\eta$ be the power-law exponent and $\rho$ be the average degree. Then:
\[ \alpha + \beta = 1 - \log(\rho) / \log(n) \quad \text{ and } \quad \alpha = \tfrac{1}{\eta - 1}. \]
We use the following treatment of the probability $p$. Suppose that the original network had $E = n \rho /2$ edges. Given the output of an MGEO-P network, we randomly delete edges until the output has exactly the same number of edges as the input network. This step can be interpreted as using the value of $p$ necessary to get the same edge count as the original graph. In the case where there are insufficient edges, we leave the output from the MGEO-P generator untouched.

\section*{Acknowledgments}
We would like to extend our thanks to MITACS for hosting our research team at the Advances in Network Analysis and its Applications Workshop held at the University of British Colombia in July 2012. Bonato and Pra{\l}at acknowledge support from NSERC DG grants. Gleich acknowledges the support of NSF CAREER award CCF-1149756.

\bibliographystyle{plain}
\bibliography{scibib}

\appendix

\section{Memoryless GEO-P Model}

\input{dimmatch-mgeop.tex}

\section{Sensitivity studies}
In the following sections, we study how the predicted dimension changes due to large scale structural changes in the graph. We focus our efforts on studying the Facebook samples as the LinkedIn samples are highly correlated due to the temporal nature of their construction.  Our results show that \begin{enumerate}
\item Erd\H{o}s R\'{e}nyi random graphs have no apparent dimension.
\item The graphlet fitting methodology is influenced by the degree distribution in a way that generates high variance in the predicted dimension but where a logarithmic trend may still exist. This effect is not present in the spectral histograms.
\item The graphlet fitting methodology is robust to changing 10\% of the edges of the network via a random percolation process.
\end{enumerate}

\subsection{Dimensions of Erd\H{o}s R\'{e}nyi random graphs}

In our first experiment to verify the relevance of our dimensionality fits, we attempt to fit the dimension of an Erd\H{o}s R\'{e}nyi random graph with the same number of expected edges. That is, for each of the samples of the Facebook network, we run the SVM dimension classifier we constructed on the graphlet counts of $50$ separate Erd\H{o}s R\'{e}nyi random graph samples where the probability is designed to yield the number of edges of the original network in expectation. In all but 3 of the 5000 examples (50 samples for each of the 100 graphs), the predicted dimension is the maximum $12$. When the dimension was not the maximum in those three cases, it was $11$. When we tried this with dimensions up to 10, then the Erd\H{o}s R\'{e}nyi random graphs fit to the dimension 10, thus, we expect these graphs to be predicted at the highest dimension of the training set. We see this as evidence that our graphlet methodology is sensitive to clearly erroneous graphs.

\subsection{Dimensions of random graphs with the same degree distribution}
In our second experiment to verify the relevance of our dimension fits, we attempt to fit the dimension of a graph with the same degree distribution as one of the Facebook networks but with edges randomly drawn. To generate these graphs, we use the Bayati-Saberi-Kim procedure\cite{Bayati} as implemented in the bisquik library\cite{Gleich}. This method terminated for 92 of the 100 graphs. (The process did not terminate in the other 8 cases, which is a limitation of this particular sampling scheme.) The dimensional fits for these 92 resampled networks are shown in Figure~\ref{fig:random-dd}. The eigenvalue fits show no logarithmic scaling in the dimension whereas the graphlet fits do. However, the variance in the predicted dimensions based on graphlets is substantially higher for these random samples compared to the original networks (see Figure 4 in the main text). The evidence from graphlets alone, is then, possibly biased due to the degree distribution. However, the results from the spectral histograms, the graphlets, and the prediction dimension from the model itself encourage us to be more optimistic.

\begin{figure}
\includegraphics[width=0.45\linewidth]{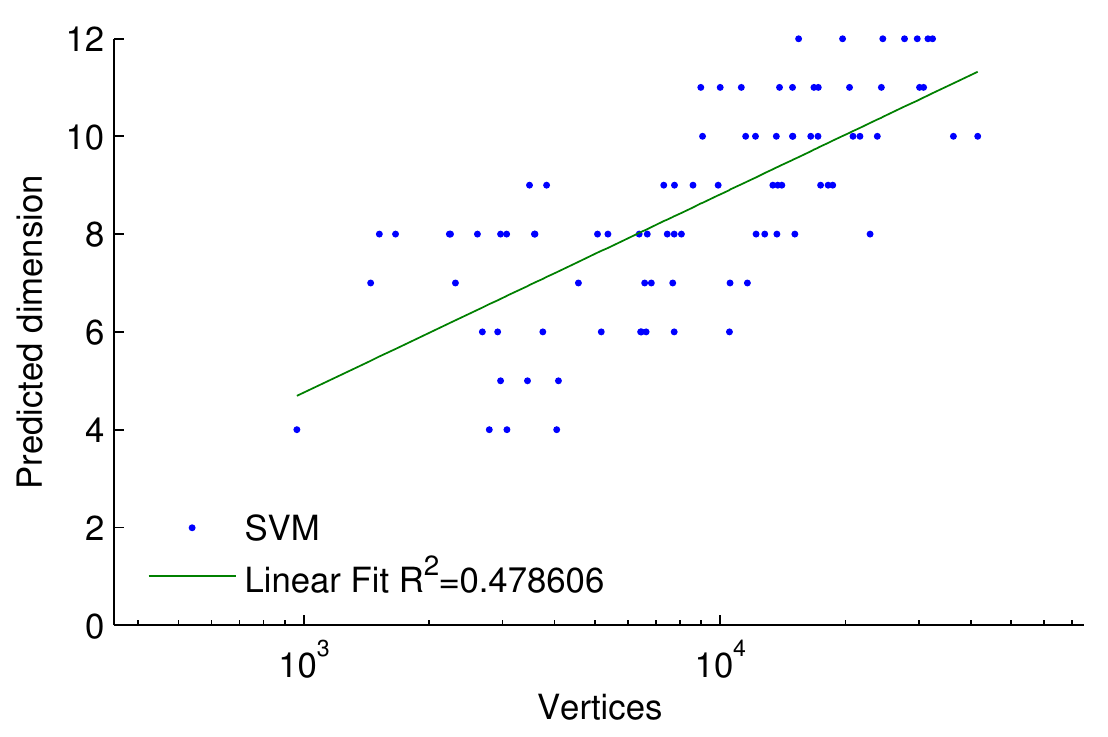}
\includegraphics[width=0.45\linewidth]{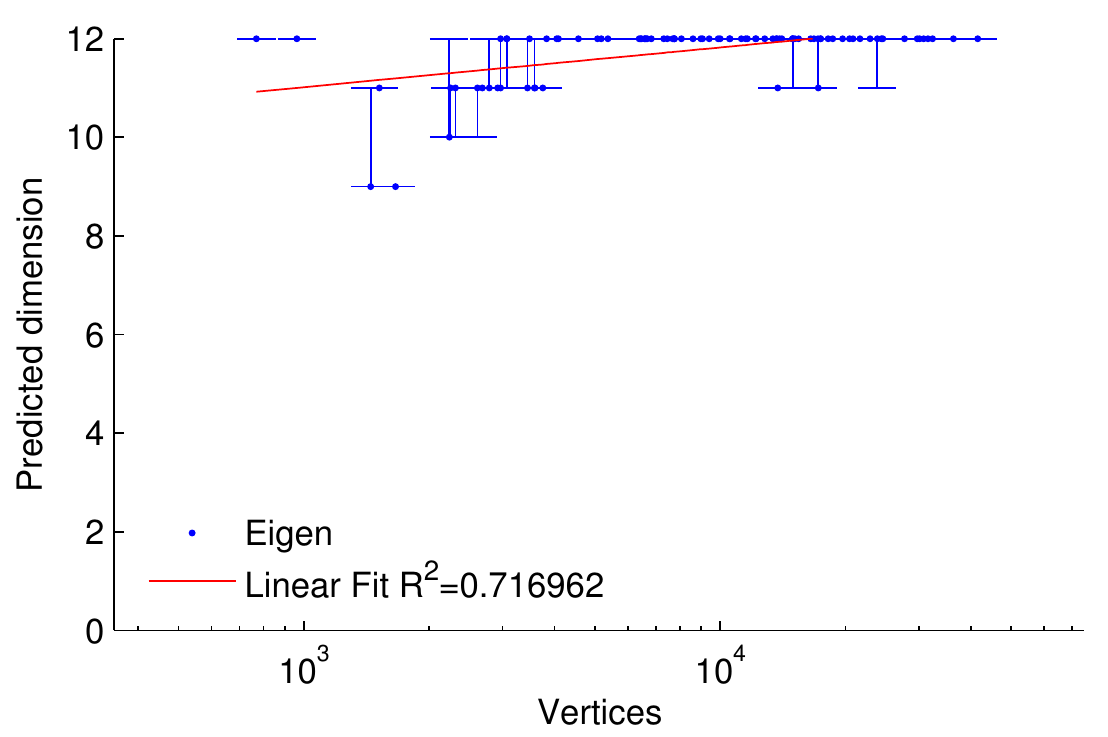}
\caption{Predicted dimensions of random graphs with the same degree distribution. At left, the predicted dimension using our SVM-Graphlet methodology and at right, the prediction dimension using our spectral histogram methodology}
\label{fig:random-dd}
\end{figure}

\subsection{Dimension variance with random percolation}
In our final experiment, we study random percolation of the predicted dimension of the Facebook networks. In a random percolation process, we randomly sample an edge from the network, delete it, replace it with an edge between two randomly drawn nodes, and continue until we have done this procedure $k$ times. We study how the predicted dimension varies as we change 1\%, 5\%, 10\%, 15\%, 20\%, 25\%, 30\%, 35\%, 40\%, 45\%, 50\% of the total edges of a network. For each of the 100 Facebook networks, and each percentage of total edges, we repeat the percolation process 10 times. This generates 110 total networks for each Facebook network. Figure~\ref{fig:perturb} shows a box-plot of how the predicted dimension varies for each perturbation level over all 1100 total graphs. This plot suggests that the dimension is unchanged until more than 15\% of the edges have been percolated. This figure further illustrates that the predicted dimension is a stable quantity for a network that is not overly sensitive to small perturbations.

\begin{figure}[!h]
\centering
\includegraphics{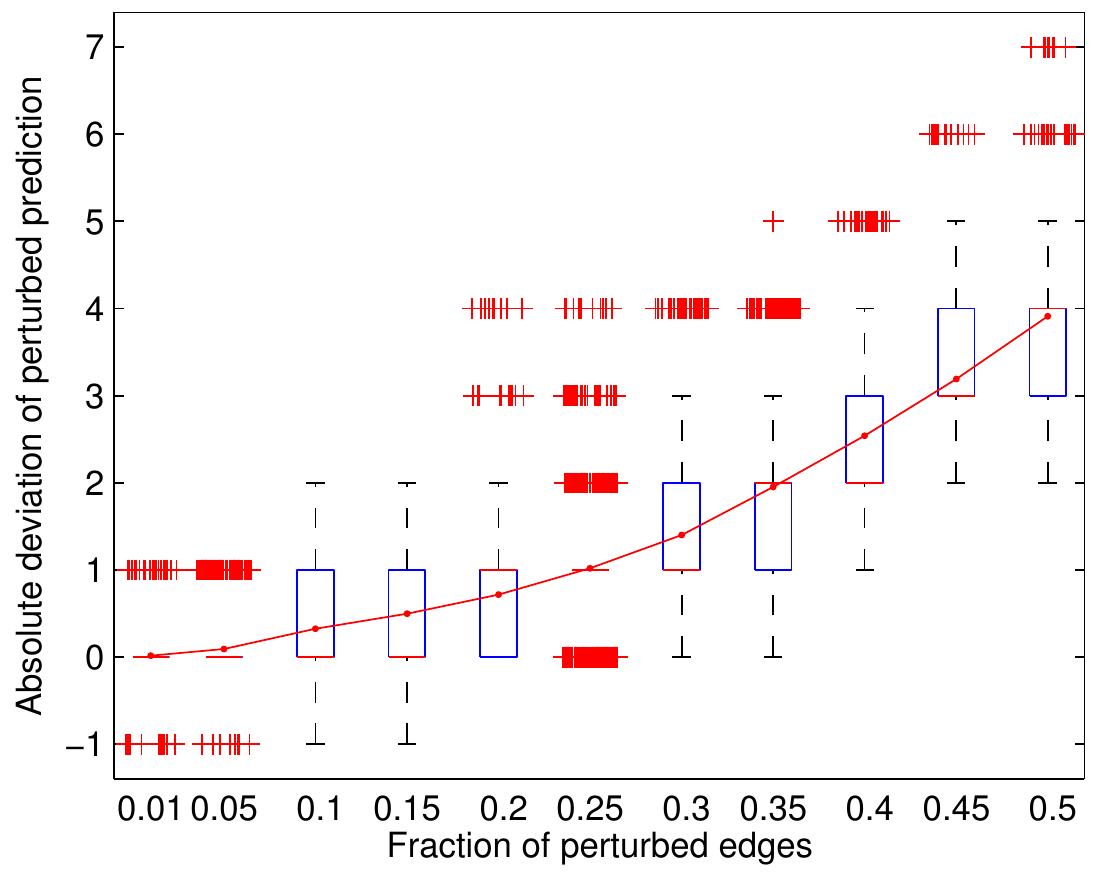}
\caption{The change in the predicted dimension based on the graphlets methodology as we randomly percolate small or large fractions of the total edges in the network. Each box-plot represents the results over all 100 Facebook networks. The label $0.05$ corresponds to randomly altering 5\% of the total edges in a network. The line tracks the mean over all the samples.}
\label{fig:perturb}
\end{figure}



\section*{Supplementary material}

# Full statistics of Facebook data

Table of properties of the Facebook networks.

| Name | Nodes | Edges | PL Exp | Eff. Diam | $\alpha$ | $\beta$ | $\text{Dim}_G$ | $\text{Dim}_E$ |
|---|---|---|---|---|---|---|---|---|
| Caltech36 | 769 | 16656 | 7.00 | 3.81 | 0.17 | 0.27 | 4 | 5 |
| Reed98 | 962 | 18812 | 4.38 | 3.88 | 0.30 | 0.17 | 3 | 5 |
| Haverford76 | 1446 | 59589 | 7.00 | 3.63 | 0.17 | 0.23 | 4 | 5 |
| Simmons81 | 1518 | 32988 | 4.74 | 3.92 | 0.27 | 0.22 | 4 | 5 |
| Swarthmore42 | 1659 | 61050 | 5.60 | 3.77 | 0.22 | 0.20 | 3 | 6 |
| Amherst41 | 2235 | 90954 | 5.64 | 3.81 | 0.22 | 0.21 | 3 | 6 |
| Bowdoin47 | 2252 | 84387 | 5.80 | 3.81 | 0.21 | 0.23 | 4 | 6 |
| Hamilton46 | 2314 | 96394 | 4.63 | 3.79 | 0.28 | 0.15 | 4 | 6 |
| Trinity100 | 2613 | 111996 | 5.83 | 3.84 | 0.21 | 0.23 | 4 | 6 |
| USFCA72 | 2682 | 65252 | 4.13 | 3.97 | 0.32 | 0.19 | 4 | 5 |
| Williams40 | 2790 | 112986 | 5.17 | 3.82 | 0.24 | 0.21 | 4 | 6 |
| Oberlin44 | 2920 | 89912 | 4.83 | 3.96 | 0.26 | 0.22 | 4 | 6 |
| Smith60 | 2970 | 97133 | 5.78 | 3.84 | 0.21 | 0.27 | 4 | 6 |
| Wellesley22 | 2970 | 94899 | 4.64 | 3.92 | 0.27 | 0.21 | 4 | 5 |
| Vassar85 | 3068 | 119161 | 5.93 | 3.82 | 0.20 | 0.26 | 4 | 6 |
| Middlebury45 | 3075 | 124610 | 6.59 | 3.92 | 0.18 | 0.27 | 4 | 7 |
| Pepperdine86 | 3445 | 152007 | 5.27 | 3.88 | 0.23 | 0.22 | 4 | 6 |
| Colgate88 | 3482 | 155043 | 6.07 | 3.81 | 0.20 | 0.25 | 4 | 6 |
| Santa74 | 3578 | 151747 | 5.74 | 3.90 | 0.21 | 0.25 | 4 | 6 |
| Wesleyan43 | 3593 | 138035 | 4.72 | 3.92 | 0.27 | 0.20 | 4 | 6 |
| Mich67 | 3748 | 81903 | 5.03 | 4.26 | 0.25 | 0.29 | 5 | 6 |
| Bucknell39 | 3826 | 158864 | 5.70 | 3.85 | 0.21 | 0.25 | 4 | 6 |
| Brandeis99 | 3898 | 137567 | 4.98 | 3.85 | 0.25 | 0.23 | 5 | 6 |
| Howard90 | 4047 | 204850 | 6.48 | 3.81 | 0.18 | 0.26 | 4 | 8 |
| Rice31 | 4087 | 184828 | 6.30 | 3.82 | 0.19 | 0.27 | 4 | 6 |
| Rochester38 | 4563 | 161404 | 5.34 | 3.96 | 0.23 | 0.27 | 5 | 6 |
| Lehigh96 | 5075 | 198347 | 6.40 | 3.90 | 0.19 | 0.30 | 4 | 6 |
| Johns-Hopkins55 | 5180 | 186586 | 5.57 | 4.07 | 0.22 | 0.28 | 5 | 6 |
| Wake73 | 5372 | 279191 | 5.71 | 3.86 | 0.21 | 0.25 | 4 | 6 |
| American75 | 6386 | 217662 | 4.85 | 4.07 | 0.26 | 0.26 | 5 | 6 |
| MIT8 | 6440 | 251252 | 5.24 | 4.02 | 0.24 | 0.27 | 5 | 6 |
| William77 | 6472 | 266378 | 5.19 | 3.90 | 0.24 | 0.26 | 5 | 6 |
| UChicago30 | 6591 | 208103 | 4.80 | 4.27 | 0.26 | 0.27 | 5 | 6 |
| Princeton12 | 6596 | 293320 | 5.35 | 4.00 | 0.23 | 0.26 | 5 | 6 |
| Carnegie49 | 6637 | 249967 | 4.98 | 4.04 | 0.25 | 0.26 | 5 | 6 |
| Tufts18 | 6682 | 249728 | 5.48 | 4.22 | 0.22 | 0.29 | 5 | 6 |
| UC64 | 6833 | 155332 | 5.71 | 4.58 | 0.21 | 0.36 | 5 | 6 |
| Vermont70 | 7324 | 191221 | 4.72 | 4.24 | 0.27 | 0.29 | 5 | 6 |
| Emory27 | 7460 | 330014 | 6.18 | 4.02 | 0.19 | 0.30 | 5 | 6 |
| Dartmouth6 | 7694 | 304076 | 5.45 | 4.11 | 0.22 | 0.29 | 5 | 6 |
| Tulane29 | 7752 | 283918 | 6.64 | 4.10 | 0.18 | 0.34 | 5 | 6 |
| WashU32 | 7755 | 367541 | 5.23 | 3.94 | 0.24 | 0.26 | 5 | 6 |
| Villanova62 | 7772 | 314989 | 5.58 | 4.09 | 0.22 | 0.29 | 5 | 6 |
| Vanderbilt48 | 8069 | 427832 | 5.61 | 3.91 | 0.22 | 0.26 | 5 | 6 |
| Yale4 | 8578 | 405450 | 5.79 | 4.02 | 0.21 | 0.29 | 5 | 6 |
| Brown11 | 8600 | 384526 | 4.85 | 4.01 | 0.26 | 0.24 | 5 | 6 |
| UCSC68 | 8991 | 224584 | 5.11 | 4.56 | 0.24 | 0.33 | 5 | 7 |
| Maine59 | 9069 | 243247 | 5.25 | 4.29 | 0.24 | 0.33 | 6 | 7 |
| Georgetown15 | 9414 | 425638 | 4.91 | 4.13 | 0.26 | 0.25 | 5 | 6 |
| Duke14 | 9895 | 506442 | 5.52 | 4.02 | 0.22 | 0.28 | 5 | 6 |
| Bingham82 | 10004 | 362894 | 5.96 | 4.08 | 0.20 | 0.33 | 5 | 7 |
| Mississippi66 | 10521 | 610911 | 5.46 | 3.90 | 0.22 | 0.26 | 5 | 6 |
| Northwestern25 | 10567 | 488337 | 5.65 | 3.98 | 0.22 | 0.30 | 5 | 6 |

| Name | Nodes | Edges | PL Exp | Eff. Diam | $\alpha$ | $\beta$ | $\text{Dim}_G$ | $\text{Dim}_E$ |
|---|---|---|---|---|---|---|---|---|
| Cal65 | 11247 | 351358 | 7.00 | 4.33 | 0.17 | 0.39 | 5 | 7 |
| BC17 | 11509 | 486967 | 5.41 | 4.13 | 0.23 | 0.30 | 5 | 7 |
| Stanford3 | 11621 | 568330 | 5.76 | 4.31 | 0.21 | 0.30 | 5 | 6 |
| Columbia2 | 11770 | 444333 | 4.29 | 4.37 | 0.30 | 0.24 | 6 | 6 |
| Notre-Dame57 | 12155 | 541339 | 4.83 | 4.01 | 0.26 | 0.26 | 6 | 7 |
| GWU54 | 12193 | 469528 | 4.76 | 4.14 | 0.27 | 0.27 | 6 | 7 |
| Baylor93 | 12803 | 679817 | 5.81 | 3.90 | 0.21 | 0.30 | 5 | 7 |
| USF51 | 13377 | 321214 | 4.93 | 4.65 | 0.25 | 0.34 | 5 | 7 |
| Syracuse56 | 13653 | 543982 | 5.93 | 4.08 | 0.20 | 0.34 | 5 | 7 |
| Temple83 | 13686 | 360795 | 4.35 | 4.52 | 0.30 | 0.29 | 6 | 7 |
| UC61 | 13746 | 442174 | 5.33 | 4.47 | 0.23 | 0.33 | 5 | 7 |
| Northeastern19 | 13882 | 381934 | 4.39 | 4.54 | 0.29 | 0.28 | 6 | 7 |
| JMU79 | 14070 | 485564 | 5.20 | 3.98 | 0.24 | 0.32 | 6 | 7 |
| UPenn7 | 14916 | 686501 | 5.89 | 4.15 | 0.20 | 0.32 | 5 | 7 |
| UCSB37 | 14935 | 482224 | 5.54 | 4.49 | 0.22 | 0.35 | 5 | 7 |
| UCF52 | 14940 | 428989 | 5.51 | 4.28 | 0.22 | 0.36 | 6 | 7 |
| UCSD34 | 14948 | 443221 | 5.01 | 4.46 | 0.25 | 0.33 | 6 | 7 |
| Harvard1 | 15126 | 824617 | 5.69 | 4.41 | 0.21 | 0.30 | 5 | 7 |
| MU78 | 15436 | 649449 | 6.21 | 4.08 | 0.19 | 0.35 | 6 | 7 |
| UMass92 | 16516 | 519385 | 4.54 | 4.15 | 0.28 | 0.29 | 6 | 7 |
| UC33 | 16808 | 522147 | 4.79 | 4.46 | 0.26 | 0.31 | 6 | 7 |
| Tennessee95 | 16979 | 770659 | 4.83 | 4.05 | 0.26 | 0.28 | 6 | 7 |
| UVA16 | 17196 | 789321 | 5.45 | 4.00 | 0.22 | 0.31 | 6 | 7 |
| UConn91 | 17212 | 604870 | 5.15 | 4.09 | 0.24 | 0.32 | 6 | 7 |
| Oklahoma97 | 17425 | 892528 | 5.26 | 3.97 | 0.23 | 0.29 | 5 | 7 |
| USC35 | 17444 | 801853 | 6.31 | 4.09 | 0.19 | 0.35 | 6 | 7 |
| UNC28 | 18163 | 766800 | 4.45 | 3.99 | 0.29 | 0.26 | 7 | 7 |
| Auburn71 | 18448 | 973918 | 4.99 | 3.92 | 0.25 | 0.28 | 6 | 7 |
| Cornell5 | 18660 | 790777 | 5.10 | 4.34 | 0.24 | 0.31 | 6 | 7 |
| BU10 | 19700 | 637528 | 5.32 | 4.49 | 0.23 | 0.35 | 6 | 7 |
| UCLA26 | 20467 | 747613 | 5.67 | 4.51 | 0.21 | 0.35 | 6 | 7 |
| Maryland58 | 20871 | 744862 | 5.41 | 4.19 | 0.23 | 0.34 | 6 | 7 |
| Virginia63 | 21325 | 698178 | 4.54 | 4.14 | 0.28 | 0.30 | 6 | 7 |
| NYU9 | 21679 | 715715 | 4.15 | 4.42 | 0.32 | 0.26 | 6 | 7 |
| Berkeley13 | 22937 | 852444 | 4.32 | 4.32 | 0.30 | 0.27 | 6 | 7 |
| Wisconsin87 | 23842 | 835952 | 4.74 | 4.30 | 0.27 | 0.31 | 6 | 7 |
| UGA50 | 24389 | 1174057 | 5.77 | 3.99 | 0.21 | 0.34 | 6 | 7 |
| Rutgers89 | 24580 | 784602 | 5.33 | 4.60 | 0.23 | 0.36 | 6 | 7 |
| FSU53 | 27737 | 1034802 | 5.79 | 4.45 | 0.21 | 0.37 | 6 | 7 |
| Indiana69 | 29747 | 1305765 | 5.40 | 4.22 | 0.23 | 0.34 | 6 | 7 |
| Michigan23 | 30147 | 1176516 | 5.12 | 4.53 | 0.24 | 0.33 | 6 | 7 |
| UIllinois20 | 30809 | 1264428 | 5.56 | 4.35 | 0.22 | 0.35 | 6 | 7 |
| Texas80 | 31560 | 1219650 | 5.65 | 4.44 | 0.22 | 0.37 | 6 | 6 |
| MSU24 | 32375 | 1118774 | 5.11 | 4.38 | 0.24 | 0.35 | 6 | 6 |
| UF21 | 35123 | 1465660 | 4.92 | 4.17 | 0.26 | 0.32 | 6 | 7 |
| Texas84 | 36371 | 1590655 | 4.79 | 4.08 | 0.26 | 0.31 | 6 | 8 |
| Penn94 | 41554 | 1362229 | 4.16 | 4.56 | 0.32 | 0.29 | 7 | 6 |

Table of log-graphlets counts of the Facebook networks

| Name | $\tilde{G}_1$ | $\tilde{G}_2$ | $\tilde{G}_3$ | $\tilde{G}_4$ | $\tilde{G}_5$ | $\tilde{G}_6$ | $\tilde{G}_7$ | $\tilde{G}_8$ |
|---|---|---|---|---|---|---|---|---|
| Caltech36 | 13.718 | 11.731 | 17.388 | 17.766 | 16.802 | 13.808 | 14.806 | 13.148 |
| Reed98 | 13.902 | 11.560 | 17.884 | 18.200 | 17.004 | 14.272 | 14.729 | 12.498 |
| Haverford76 | 15.514 | 13.307 | 19.375 | 20.127 | 18.868 | 16.273 | 16.775 | 14.728 |

| Name | $\tilde{G}_1$ | $\tilde{G}_2$ | $\tilde{G}_3$ | $\tilde{G}_4$ | $\tilde{G}_5$ | $\tilde{G}_6$ | $\tilde{G}_7$ | $\tilde{G}_8$ |
|---|---|---|---|---|---|---|---|---|
| Simmons81 | 14.421 | 12.015 | 18.476 | 18.950 | 17.604 | 14.759 | 15.341 | 13.296 |
| Swarthmore42 | 15.558 | 13.229 | 19.853 | 20.360 | 19.099 | 16.372 | 16.884 | 14.733 |
| Amherst41 | 15.964 | 13.672 | 20.383 | 20.974 | 19.725 | 16.908 | 17.565 | 15.582 |
| Bowdoin47 | 15.868 | 13.491 | 20.066 | 20.694 | 19.302 | 16.519 | 17.053 | 15.091 |
| Hamilton46 | 16.075 | 13.724 | 20.384 | 20.950 | 19.552 | 16.808 | 17.239 | 15.129 |
| Trinity100 | 16.201 | 13.887 | 20.534 | 21.243 | 19.833 | 17.065 | 17.596 | 15.634 |
| USFCA72 | 15.404 | 12.859 | 19.837 | 20.414 | 18.884 | 16.018 | 16.458 | 14.442 |
| Williams40 | 16.270 | 13.833 | 20.883 | 21.393 | 20.029 | 17.129 | 17.696 | 15.621 |
| Oberlin44 | 15.860 | 13.215 | 20.484 | 20.975 | 19.351 | 16.442 | 16.813 | 14.691 |
| Smith60 | 15.860 | 13.325 | 20.315 | 20.981 | 19.411 | 16.281 | 16.912 | 15.261 |
| Wellesley22 | 15.993 | 13.350 | 20.594 | 21.118 | 19.533 | 16.633 | 16.993 | 14.840 |
| Vassar85 | 16.294 | 13.655 | 20.725 | 21.409 | 19.774 | 17.012 | 17.262 | 15.029 |
| Middlebury45 | 16.338 | 13.914 | 20.785 | 21.454 | 19.995 | 17.117 | 17.672 | 15.709 |
| Pepperdine86 | 16.763 | 14.319 | 21.290 | 21.864 | 20.403 | 17.543 | 18.043 | 16.114 |
| Colgate88 | 16.570 | 14.120 | 20.940 | 21.675 | 20.159 | 17.317 | 17.827 | 15.916 |
| Santa74 | 16.642 | 14.183 | 21.265 | 21.854 | 20.362 | 17.468 | 17.996 | 16.035 |
| Wesleyan43 | 16.523 | 14.002 | 20.967 | 21.619 | 20.035 | 17.142 | 17.618 | 15.653 |
| Mich67 | 15.523 | 12.990 | 20.435 | 20.894 | 19.409 | 16.214 | 16.944 | 15.171 |
| Bucknell39 | 16.604 | 14.135 | 21.011 | 21.762 | 20.236 | 17.303 | 17.879 | 16.079 |
| Brandeis99 | 16.584 | 13.875 | 23.371 | 21.924 | 21.023 | 17.326 | 18.008 | 15.722 |
| Howard90 | 17.223 | 14.542 | 22.062 | 22.623 | 20.957 | 18.341 | 18.425 | 16.117 |
| Rice31 | 16.950 | 14.489 | 21.533 | 22.186 | 20.657 | 17.556 | 18.237 | 16.396 |
| Rochester38 | 16.592 | 14.103 | 21.867 | 21.935 | 20.429 | 17.335 | 17.966 | 16.241 |
| Lehigh96 | 16.849 | 14.300 | 21.824 | 22.269 | 20.725 | 17.709 | 18.273 | 16.449 |
| Johns-Hopkins55 | 16.825 | 14.287 | 22.032 | 22.336 | 20.841 | 17.745 | 18.391 | 16.594 |
| Wake73 | 17.496 | 15.025 | 22.285 | 22.886 | 21.320 | 18.248 | 18.895 | 17.249 |
| American75 | 16.981 | 14.238 | 21.921 | 22.521 | 20.751 | 17.591 | 18.095 | 16.398 |
| MIT8 | 17.263 | 14.629 | 22.254 | 22.856 | 21.200 | 18.098 | 18.656 | 16.864 |
| William77 | 17.297 | 14.587 | 23.076 | 23.038 | 21.438 | 18.222 | 18.679 | 16.888 |
| UChicago30 | 16.832 | 14.068 | 22.312 | 22.533 | 20.798 | 17.599 | 18.102 | 16.351 |
| Princeton12 | 17.450 | 14.720 | 22.371 | 23.056 | 21.296 | 18.367 | 18.714 | 16.637 |
| Carnegie49 | 17.195 | 14.619 | 22.400 | 22.890 | 21.264 | 18.152 | 18.713 | 16.857 |
| Tufts18 | 17.178 | 14.458 | 22.155 | 22.758 | 20.978 | 17.908 | 18.353 | 16.427 |
| UC64 | 16.320 | 13.797 | 21.212 | 21.792 | 20.205 | 16.810 | 17.797 | 16.375 |
| Vermont70 | 16.550 | 13.750 | 21.667 | 22.193 | 20.403 | 17.205 | 17.741 | 15.907 |
| Emory27 | 17.508 | 14.937 | 22.487 | 23.133 | 21.521 | 18.507 | 19.140 | 17.365 |
| Dartmouth6 | 17.437 | 14.626 | 22.655 | 23.152 | 21.351 | 18.337 | 18.757 | 16.699 |
| Tulane29 | 17.287 | 14.752 | 22.379 | 22.950 | 21.338 | 18.099 | 18.891 | 17.347 |
| WashU32 | 17.751 | 15.070 | 24.290 | 23.744 | 22.343 | 18.961 | 19.492 | 17.472 |
| Villanova62 | 17.474 | 14.767 | 22.473 | 23.104 | 21.376 | 18.316 | 18.806 | 16.948 |
| Vanderbilt48 | 18.003 | 15.401 | 23.641 | 23.678 | 22.035 | 18.877 | 19.431 | 17.748 |
| Yale4 | 17.867 | 15.043 | 23.511 | 23.710 | 21.951 | 18.763 | 19.200 | 17.250 |
| Brown11 | 17.739 | 14.846 | 23.016 | 23.580 | 21.716 | 18.701 | 19.047 | 17.014 |
| UCSC68 | 16.615 | 13.980 | 21.546 | 22.232 | 20.499 | 17.054 | 17.956 | 16.316 |
| Maine59 | 16.843 | 13.954 | 22.256 | 22.519 | 20.624 | 17.278 | 17.875 | 16.161 |
| Georgetown15 | 17.875 | 15.029 | 23.093 | 23.689 | 21.824 | 18.778 | 19.132 | 17.152 |
| Duke14 | 18.194 | 15.468 | 23.419 | 24.008 | 22.262 | 19.163 | 19.713 | 17.913 |
| Bingham82 | 17.416 | 14.647 | 22.543 | 23.260 | 21.441 | 18.197 | 18.915 | 17.258 |
| Mississippi66 | 18.525 | 15.924 | 24.072 | 24.494 | 22.847 | 19.665 | 20.335 | 18.567 |
| Northwestern25 | 18.054 | 15.302 | 24.403 | 24.064 | 22.399 | 19.015 | 19.570 | 17.727 |
| Cal65 | 17.326 | 14.567 | 22.509 | 23.209 | 21.433 | 17.982 | 18.757 | 17.390 |
| BC17 | 17.905 | 14.992 | 23.342 | 23.902 | 22.019 | 18.969 | 19.366 | 17.386 |
| Stanford3 | 18.314 | 15.527 | 23.744 | 24.356 | 22.520 | 19.306 | 19.783 | 18.121 |
| Columbia2 | 18.027 | 15.042 | 24.979 | 24.149 | 22.584 | 19.149 | 19.705 | 17.484 |
| Notre-Dame57 | 18.120 | 15.088 | 23.689 | 24.111 | 22.105 | 18.935 | 19.317 | 17.287 |
| GWU54 | 17.944 | 15.032 | 24.613 | 24.060 | 22.284 | 18.838 | 19.363 | 17.599 |
| Baylor93 | 18.529 | 15.734 | 24.456 | 24.655 | 22.822 | 19.610 | 20.119 | 18.376 |
| USF51 | 17.259 | 14.443 | 22.687 | 23.190 | 21.281 | 18.062 | 18.569 | 16.929 |

| Name | $\tilde{G}_1$ | $\tilde{G}_2$ | $\tilde{G}_3$ | $\tilde{G}_4$ | $\tilde{G}_5$ | $\tilde{G}_6$ | $\tilde{G}_7$ | $\tilde{G}_8$ |
| --- | --- | --- | --- | --- | --- | --- | --- | --- |
| Syracuse56 | 17.948 | 15.270 | 23.999 | 24.101 | 22.329 | 18.987 | 19.663 | 18.297 |
| Temple83 | 17.470 | 14.446 | 23.621 | 23.581 | 21.669 | 18.532 | 18.860 | 16.869 |
| UC61 | 17.692 | 15.061 | 23.205 | 23.766 | 22.100 | 18.442 | 19.376 | 17.989 |
| Northeastern19 | 17.324 | 14.343 | 22.931 | 23.363 | 21.360 | 18.038 | 18.618 | 16.892 |
| JMU79 | 17.826 | 14.817 | 25.055 | 23.911 | 22.157 | 18.468 | 19.027 | 17.504 |
| UPenn7 | 18.403 | 15.520 | 23.941 | 24.550 | 22.617 | 19.369 | 19.867 | 18.070 |
| UCSB37 | 17.768 | 14.983 | 23.076 | 23.754 | 21.915 | 18.352 | 19.210 | 17.784 |
| UCF52 | 17.936 | 15.121 | 25.126 | 23.785 | 22.209 | 18.435 | 19.460 | 18.270 |
| UCSD34 | 17.596 | 14.777 | 23.748 | 23.679 | 21.847 | 18.195 | 19.061 | 17.590 |
| Harvard1 | 18.857 | 15.913 | 24.389 | 25.025 | 23.095 | 20.079 | 20.333 | 18.188 |
| MU78 | 18.186 | 15.370 | 23.543 | 24.261 | 22.387 | 18.844 | 19.641 | 18.263 |
| UMass92 | 17.800 | 14.756 | 24.903 | 23.986 | 22.100 | 18.601 | 19.198 | 17.446 |
| UC33 | 17.831 | 14.971 | 23.489 | 23.945 | 22.109 | 18.467 | 19.340 | 17.816 |
| Tennessee95 | 18.822 | 15.821 | 25.122 | 24.846 | 22.964 | 19.566 | 20.153 | 18.553 |
| UVA16 | 18.681 | 15.691 | 24.568 | 24.786 | 22.854 | 19.422 | 20.060 | 18.381 |
| UConn91 | 18.012 | 15.073 | 23.583 | 24.141 | 22.113 | 18.638 | 19.347 | 17.725 |
| Oklahoma97 | 18.956 | 16.166 | 24.691 | 25.130 | 23.322 | 19.898 | 20.686 | 19.128 |
| USC35 | 18.681 | 15.760 | 24.776 | 24.863 | 22.972 | 19.538 | 20.228 | 18.667 |
| UNC28 | 18.724 | 15.544 | 24.894 | 24.532 | 22.325 | 18.877 | 19.196 | 17.272 |
| Auburn71 | 19.105 | 16.140 | 25.604 | 25.360 | 23.539 | 19.947 | 20.678 | 19.123 |
| Cornell5 | 18.598 | 15.614 | 25.378 | 24.834 | 22.939 | 19.449 | 20.112 | 18.464 |
| BU10 | 18.035 | 14.925 | 23.718 | 24.231 | 22.137 | 18.678 | 19.259 | 17.708 |
| UCLA26 | 18.331 | 15.435 | 23.907 | 24.597 | 22.657 | 19.068 | 19.866 | 18.416 |
| Maryland58 | 18.376 | 15.374 | 25.593 | 24.728 | 22.783 | 19.165 | 19.773 | 18.264 |
| Virginia63 | 18.448 | 15.322 | 26.598 | 24.835 | 23.067 | 19.158 | 19.911 | 18.418 |
| NYU9 | 18.307 | 15.117 | 24.927 | 24.656 | 22.562 | 19.064 | 19.574 | 17.772 |
| Berkeley13 | 18.637 | 15.458 | 26.726 | 25.366 | 23.590 | 19.668 | 20.406 | 18.426 |
| Wisconsin87 | 18.472 | 15.374 | 25.177 | 24.862 | 22.817 | 19.153 | 19.957 | 18.362 |
| UGA50 | 19.010 | 16.124 | 24.970 | 25.418 | 23.515 | 19.869 | 20.751 | 19.238 |
| Rutgers89 | 18.244 | 15.303 | 24.316 | 24.585 | 22.615 | 18.987 | 19.843 | 18.299 |
| FSU53 | 18.729 | 15.938 | 24.706 | 25.089 | 23.220 | 19.551 | 20.543 | 19.192 |
| Indiana69 | 18.989 | 16.028 | 24.958 | 25.537 | 23.496 | 19.971 | 20.731 | 19.257 |
| Michigan23 | 18.861 | 15.914 | 25.187 | 25.433 | 23.418 | 19.984 | 20.730 | 19.046 |
| UIllinois20 | 18.957 | 16.040 | 25.801 | 25.446 | 23.498 | 19.865 | 20.728 | 19.296 |
| Texas80 | 18.905 | 16.092 | 24.841 | 25.387 | 23.520 | 19.807 | 20.845 | 19.340 |
| MSU24 | 18.829 | 15.676 | 25.554 | 25.234 | 23.194 | 19.711 | 20.328 | 18.660 |
| UF21 | 19.324 | 16.295 | 27.899 | 26.069 | 24.224 | 20.261 | 21.154 | 19.737 |
| Texas84 | 19.574 | 16.232 | 26.785 | 26.199 | 24.047 | 20.229 | 20.911 | 19.327 |
| Penn94 | 19.106 | 15.797 | 26.853 | 25.886 | 23.738 | 19.985 | 20.750 | 18.978 |

# Full statistics of LinkedIn data

Table of properties of the LinkedIn networks. We only compute eigenvalue dimensional fits up to $72,000$ nodes.

| Nodes | Edges | PL Exp | Eff. Diam | $\alpha$ | $\beta$ | $\text{Dim}_G$ | $\text{Dim}_E$ |
|---|---|---|---|---|---|---|---|
| 578 | 3140 | 3.01 | 4.70 | 0.50 | 0.14 | 4 | 3 |
| 650 | 3596 | 3.02 | 4.75 | 0.50 | 0.13 | 4 | 4 |
| 709 | 4032 | 2.93 | 4.68 | 0.52 | 0.12 | 3 | 3 |
| 774 | 4448 | 2.77 | 4.68 | 0.56 | 0.07 | 3 | 3 |
| 916 | 5300 | 2.74 | 4.71 | 0.57 | 0.07 | 3 | 3 |
| 1012 | 5983 | 2.82 | 4.69 | 0.55 | 0.10 | 4 | 4 |
| 1158 | 6955 | 2.77 | 4.77 | 0.56 | 0.08 | 4 | 4 |
| 1312 | 7957 | 2.75 | 4.79 | 0.57 | 0.08 | 4 | 3 |
| 1512 | 9287 | 2.80 | 4.83 | 0.56 | 0.11 | 4 | 4 |
| 1800 | 11366 | 2.73 | 4.93 | 0.58 | 0.09 | 4 | 3 |
| 2259 | 14643 | 2.66 | 5.07 | 0.60 | 0.08 | 4 | 3 |
| 2643 | 17795 | 2.74 | 5.06 | 0.57 | 0.10 | 4 | 4 |
| 2972 | 20465 | 2.61 | 5.10 | 0.62 | 0.06 | 4 | 3 |
| 3389 | 23432 | 2.76 | 5.20 | 0.57 | 0.12 | 4 | 4 |
| 3817 | 26656 | 2.82 | 5.28 | 0.55 | 0.14 | 5 | 4 |
| 4158 | 29086 | 2.81 | 5.23 | 0.55 | 0.14 | 4 | 4 |
| 4508 | 31557 | 2.80 | 5.30 | 0.56 | 0.13 | 4 | 4 |
| 4898 | 34273 | 2.81 | 5.37 | 0.55 | 0.15 | 5 | 4 |
| 5288 | 37173 | 2.85 | 5.30 | 0.54 | 0.15 | 5 | 4 |
| 5728 | 40550 | 2.86 | 5.36 | 0.54 | 0.16 | 5 | 4 |
| 6216 | 44091 | 2.85 | 5.40 | 0.54 | 0.16 | 4 | 4 |
| 6719 | 47813 | 2.81 | 5.47 | 0.55 | 0.15 | 5 | 4 |
| 7230 | 51600 | 2.84 | 5.50 | 0.54 | 0.16 | 5 | 4 |
| 7892 | 56238 | 2.83 | 5.47 | 0.55 | 0.16 | 5 | 4 |
| 8551 | 60985 | 2.87 | 5.56 | 0.53 | 0.17 | 5 | 4 |
| 9572 | 68113 | 2.88 | 5.54 | 0.53 | 0.18 | 5 | 4 |
| 10444 | 74215 | 2.90 | 5.55 | 0.53 | 0.19 | 5 | 4 |
| 11222 | 79805 | 2.91 | 5.58 | 0.52 | 0.19 | 5 | 4 |
| 12210 | 87129 | 2.93 | 5.57 | 0.52 | 0.20 | 5 | 5 |
| 13043 | 93158 | 2.95 | 5.59 | 0.51 | 0.21 | 5 | 5 |
| 14173 | 101279 | 3.22 | 5.59 | 0.45 | 0.27 | 6 | 5 |
| 15246 | 109186 | 2.94 | 5.59 | 0.52 | 0.21 | 5 | 4 |
| 16499 | 118284 | 3.22 | 5.62 | 0.45 | 0.28 | 6 | 5 |
| 17816 | 127765 | 3.23 | 5.65 | 0.45 | 0.28 | 5 | 5 |
| 19522 | 140191 | 3.28 | 5.67 | 0.44 | 0.29 | 6 | 5 |
| 21125 | 151651 | 2.79 | 5.68 | 0.56 | 0.18 | 6 | 4 |
| 22665 | 163457 | 2.79 | 5.66 | 0.56 | 0.18 | 5 | 4 |
| 24530 | 177418 | 3.38 | 5.69 | 0.42 | 0.32 | 7 | 5 |
| 26535 | 191947 | 2.87 | 5.72 | 0.53 | 0.21 | 5 | 4 |
| 28760 | 207848 | 3.78 | 5.72 | 0.36 | 0.38 | 7 | 5 |
| 31176 | 225897 | 2.84 | 5.77 | 0.54 | 0.20 | 6 | 4 |
| 33834 | 246160 | 2.84 | 5.79 | 0.54 | 0.20 | 5 | 4 |
| 37010 | 270718 | 3.39 | 5.78 | 0.42 | 0.33 | 6 | 5 |
| 40032 | 295219 | 3.34 | 5.80 | 0.43 | 0.32 | 6 | 5 |
| 43567 | 322710 | 2.78 | 5.81 | 0.56 | 0.19 | 5 | 4 |
| 47027 | 350726 | 2.79 | 5.83 | 0.56 | 0.20 | 5 | 4 |
| 51542 | 387570 | 2.78 | 5.81 | 0.56 | 0.19 | 5 | 4 |
| 55984 | 423919 | 3.39 | 5.82 | 0.42 | 0.33 | 6 | 5 |
| 61138 | 465192 | 3.44 | 5.83 | 0.41 | 0.34 | 6 | 5 |
| 66622 | 510201 | 3.52 | 5.85 | 0.40 | 0.36 | 6 | 5 |
| 71627 | 551151 | 3.48 | 5.85 | 0.40 | 0.35 | 5 | 5 |

| Nodes | Edges | PL Exp | Eff. Diam | $\alpha$ | $\beta$ | $\text{Dim}_G$ | $\text{Dim}_E$ |
|---|---|---|---|---|---|---|---|
| 77093 | 594813 | 3.51 | 5.85 | 0.40 | 0.36 | 7 | – |
| 82869 | 640416 | 3.48 | 5.84 | 0.40 | 0.36 | 7 | – |
| 88870 | 688431 | 3.46 | 5.84 | 0.41 | 0.36 | 7 | – |
| 95073 | 738405 | 3.44 | 5.85 | 0.41 | 0.35 | 7 | – |
| 102902 | 802009 | 3.42 | 5.86 | 0.41 | 0.35 | 7 | – |
| 111651 | 874290 | 3.33 | 5.86 | 0.43 | 0.34 | 6 | – |
| 122320 | 963418 | 3.34 | 5.85 | 0.43 | 0.34 | 6 | – |
| 132572 | 1051233 | 3.30 | 5.83 | 0.43 | 0.34 | 6 | – |
| 142230 | 1131895 | 3.28 | 5.87 | 0.44 | 0.33 | 5 | – |
| 152676 | 1219866 | 3.11 | 5.85 | 0.47 | 0.30 | 6 | – |
| 164333 | 1322238 | 3.11 | 5.86 | 0.47 | 0.30 | 6 | – |
| 177058 | 1436964 | 3.13 | 5.85 | 0.47 | 0.30 | 6 | – |
| 193345 | 1582671 | 3.11 | 5.85 | 0.47 | 0.30 | 6 | – |
| 209628 | 1732864 | 3.10 | 5.83 | 0.48 | 0.30 | 6 | – |
| 226563 | 1886060 | 3.11 | 5.85 | 0.47 | 0.30 | 6 | – |
| 245398 | 2058727 | 3.11 | 5.85 | 0.47 | 0.30 | 6 | – |
| 354977 | 3184846 | 3.06 | 5.78 | 0.49 | 0.29 | 6 | – |
| 397649 | 3675170 | 3.05 | 5.77 | 0.49 | 0.29 | 6 | – |
| 441553 | 4214245 | 3.00 | 5.75 | 0.50 | 0.27 | 5 | – |
| 489997 | 4827461 | 2.96 | 5.71 | 0.51 | 0.27 | 5 | – |

Table of log-graphlets counts of the LinkedIn networks

| Nodes | Edges | $\tilde{G}_1$ | $\tilde{G}_2$ | $\tilde{G}_3$ | $\tilde{G}_4$ | $\tilde{G}_5$ | $\tilde{G}_6$ | $\tilde{G}_7$ | $\tilde{G}_8$ |
|---|---|---|---|---|---|---|---|---|---|
| 578 | 3140 | 10.678 | 8.080 | 14.270 | 14.010 | 12.623 | 9.442 | 10.268 | 8.009 |
| 650 | 3596 | 11.307 | 8.431 | 14.360 | 14.339 | 12.850 | 9.630 | 10.333 | 8.593 |
| 709 | 4032 | 10.984 | 8.191 | 14.744 | 14.447 | 12.969 | 9.860 | 10.536 | 8.466 |
| 774 | 4448 | 11.127 | 8.343 | 14.544 | 14.567 | 12.975 | 9.870 | 10.482 | 8.281 |
| 916 | 5300 | 11.545 | 8.548 | 15.146 | 15.140 | 13.552 | 10.558 | 10.973 | 8.324 |
| 1012 | 5983 | 11.679 | 8.627 | 15.359 | 15.352 | 13.729 | 10.604 | 11.147 | 9.139 |
| 1158 | 6955 | 11.975 | 8.907 | 15.813 | 15.704 | 14.067 | 11.013 | 11.402 | 8.895 |
| 1312 | 7957 | 12.037 | 8.987 | 16.088 | 15.968 | 14.128 | 11.117 | 11.306 | 8.937 |
| 1512 | 9287 | 11.981 | 8.910 | 16.181 | 16.124 | 14.272 | 11.245 | 11.480 | 9.044 |
| 1800 | 11366 | 12.312 | 9.220 | 16.612 | 16.586 | 14.704 | 11.626 | 11.908 | 9.624 |
| 2259 | 14643 | 12.829 | 9.611 | 17.055 | 16.942 | 14.908 | 11.793 | 12.065 | 9.777 |
| 2643 | 17795 | 12.950 | 9.818 | 17.430 | 17.355 | 15.410 | 12.197 | 12.603 | 10.141 |
| 2972 | 20465 | 13.370 | 10.055 | 17.753 | 17.616 | 15.654 | 12.515 | 12.816 | 10.289 |
| 3389 | 23432 | 13.400 | 10.158 | 18.100 | 17.857 | 15.849 | 12.669 | 12.935 | 10.441 |
| 3817 | 26656 | 13.637 | 10.429 | 18.340 | 18.168 | 16.176 | 12.856 | 13.198 | 10.994 |
| 4158 | 29086 | 13.756 | 10.457 | 18.462 | 18.271 | 16.198 | 12.879 | 13.258 | 10.979 |
| 4508 | 31557 | 13.595 | 10.275 | 18.466 | 18.351 | 16.277 | 13.087 | 13.277 | 10.947 |
| 4898 | 34273 | 13.938 | 10.669 | 18.724 | 18.549 | 16.476 | 13.110 | 13.424 | 11.171 |
| 5288 | 37173 | 13.934 | 10.588 | 18.907 | 18.709 | 16.554 | 13.301 | 13.486 | 11.050 |
| 5728 | 40550 | 14.019 | 10.704 | 18.863 | 18.782 | 16.633 | 13.360 | 13.633 | 11.138 |
| 6216 | 44091 | 14.151 | 10.861 | 19.020 | 18.906 | 16.768 | 13.440 | 13.950 | 11.909 |
| 6719 | 47813 | 14.166 | 10.793 | 19.205 | 19.043 | 16.876 | 13.490 | 13.894 | 11.573 |
| 7230 | 51600 | 14.211 | 10.757 | 19.294 | 19.138 | 16.985 | 13.604 | 14.056 | 11.973 |
| 7892 | 56238 | 14.425 | 11.036 | 19.416 | 19.290 | 17.054 | 13.675 | 14.063 | 11.819 |
| 8551 | 60985 | 14.603 | 11.195 | 19.590 | 19.467 | 17.200 | 13.877 | 14.183 | 12.030 |
| 9572 | 68113 | 14.601 | 11.177 | 19.731 | 19.616 | 17.340 | 14.060 | 14.324 | 12.466 |
| 10444 | 74215 | 14.714 | 11.250 | 19.884 | 19.776 | 17.488 | 14.081 | 14.538 | 12.177 |
| 11222 | 79805 | 14.846 | 11.323 | 20.062 | 19.938 | 17.660 | 14.307 | 14.618 | 12.336 |
| 12210 | 87129 | 14.859 | 11.328 | 20.262 | 20.071 | 17.785 | 14.374 | 14.815 | 12.742 |
| 13043 | 93158 | 14.902 | 11.391 | 20.347 | 20.150 | 17.842 | 14.340 | 14.886 | 12.793 |

| Nodes | Edges | $\tilde{G}_1$ | $\tilde{G}_2$ | $\tilde{G}_3$ | $\tilde{G}_4$ | $\tilde{G}_5$ | $\tilde{G}_6$ | $\tilde{G}_7$ | $\tilde{G}_8$ |
|---|---|---|---|---|---|---|---|---|---|
| 14173 | 101279 | 15.075 | 11.475 | 20.485 | 20.336 | 17.975 | 14.590 | 14.979 | 12.737 |
| 15246 | 109186 | 15.192 | 11.585 | 20.530 | 20.391 | 17.970 | 14.588 | 14.910 | 12.745 |
| 16499 | 118284 | 15.287 | 11.706 | 20.780 | 20.561 | 18.211 | 14.738 | 15.147 | 12.942 |
| 17816 | 127765 | 15.348 | 11.718 | 20.896 | 20.700 | 18.365 | 14.808 | 15.328 | 13.113 |
| 19522 | 140191 | 15.407 | 11.786 | 20.997 | 20.886 | 18.490 | 15.042 | 15.476 | 13.284 |
| 21125 | 151651 | 15.574 | 11.892 | 21.147 | 21.001 | 18.584 | 15.123 | 15.520 | 13.240 |
| 22665 | 163457 | 15.545 | 11.894 | 21.324 | 21.098 | 18.692 | 15.193 | 15.657 | 13.372 |
| 24530 | 177418 | 15.695 | 12.010 | 21.394 | 21.229 | 18.813 | 15.341 | 15.775 | 13.531 |
| 26535 | 191947 | 15.856 | 12.175 | 21.575 | 21.375 | 18.919 | 15.421 | 15.855 | 13.797 |
| 28760 | 207848 | 15.901 | 12.218 | 21.773 | 21.496 | 19.061 | 15.557 | 16.028 | 13.798 |
| 31176 | 225897 | 16.062 | 12.383 | 21.861 | 21.642 | 19.204 | 15.681 | 16.232 | 14.142 |
| 33834 | 246160 | 16.121 | 12.431 | 21.976 | 21.751 | 19.328 | 15.833 | 16.457 | 14.196 |
| 37010 | 270718 | 16.109 | 12.442 | 22.189 | 21.905 | 19.480 | 15.996 | 16.560 | 14.376 |
| 40032 | 295219 | 16.308 | 12.580 | 22.390 | 22.087 | 19.681 | 16.157 | 16.758 | 14.438 |
| 43567 | 322710 | 16.311 | 12.661 | 22.513 | 22.241 | 19.883 | 16.388 | 17.021 | 14.667 |
| 47027 | 350726 | 16.447 | 12.772 | 22.910 | 22.449 | 20.115 | 16.625 | 17.254 | 14.909 |
| 51542 | 387570 | 16.643 | 12.907 | 23.111 | 22.622 | 20.307 | 16.800 | 17.468 | 15.314 |
| 55984 | 423919 | 16.672 | 13.042 | 23.136 | 22.703 | 20.462 | 16.986 | 17.686 | 15.354 |
| 61138 | 465192 | 16.785 | 13.091 | 23.437 | 22.882 | 20.627 | 17.095 | 17.831 | 15.562 |
| 66622 | 510201 | 16.909 | 13.252 | 23.495 | 23.050 | 20.757 | 17.263 | 17.995 | 15.727 |
| 71627 | 551151 | 17.021 | 13.327 | 23.814 | 23.221 | 20.953 | 17.420 | 18.241 | 15.881 |
| 77093 | 594813 | 17.326 | 13.518 | 24.178 | 23.393 | 21.177 | 17.532 | 18.447 | 16.081 |
| 82869 | 640416 | 17.171 | 13.494 | 24.223 | 23.509 | 21.307 | 17.627 | 18.560 | 16.126 |
| 88870 | 688431 | 17.360 | 13.533 | 24.481 | 23.671 | 21.411 | 17.763 | 18.696 | 16.325 |
| 95073 | 738405 | 17.335 | 13.648 | 24.491 | 23.778 | 21.467 | 17.780 | 18.650 | 16.217 |
| 102902 | 802009 | 17.494 | 13.791 | 24.777 | 24.000 | 21.704 | 17.954 | 18.901 | 16.598 |
| 111651 | 874290 | 17.561 | 13.771 | 24.827 | 24.137 | 21.795 | 18.075 | 18.979 | 16.630 |
| 122320 | 963418 | 17.725 | 13.903 | 25.348 | 24.403 | 22.081 | 18.252 | 19.307 | 16.934 |
| 132572 | 1051233 | 17.969 | 14.055 | 25.645 | 24.607 | 22.280 | 18.363 | 19.539 | 17.232 |
| 142230 | 1131895 | 18.002 | 14.160 | 25.824 | 24.723 | 22.382 | 18.438 | 19.612 | 17.297 |
| 152676 | 1219866 | 18.035 | 14.195 | 26.271 | 24.898 | 22.631 | 18.627 | 19.866 | 17.313 |
| 164333 | 1322238 | 18.129 | 14.247 | 26.387 | 25.058 | 22.757 | 18.709 | 19.881 | 17.430 |
| 177058 | 1436964 | 18.531 | 14.481 | 26.716 | 25.250 | 22.942 | 18.814 | 20.034 | 17.458 |
| 193345 | 1582671 | 18.459 | 14.531 | 26.974 | 25.471 | 23.162 | 19.039 | 20.179 | 17.742 |
| 209628 | 1732864 | 18.681 | 14.664 | 27.251 | 25.644 | 23.265 | 19.148 | 20.219 | 17.881 |
| 226563 | 1886060 | 18.917 | 14.762 | 27.405 | 25.809 | 23.464 | 19.330 | 20.427 | 18.021 |
| 245398 | 2058727 | 18.849 | 14.867 | 27.594 | 25.941 | 23.604 | 19.391 | 20.522 | 18.135 |
| 354977 | 3184846 | 19.350 | 15.411 | 28.372 | 26.869 | 24.554 | 20.246 | 21.400 | 19.064 |
| 397649 | 3675170 | 19.705 | 15.624 | 28.791 | 27.215 | 24.904 | 20.512 | 21.733 | 19.452 |
| 441553 | 4214245 | 20.009 | 16.003 | 29.337 | 27.588 | 25.406 | 20.909 | 22.201 | 19.979 |
| 489997 | 4827461 | 20.415 | 16.301 | 29.879 | 28.002 | 25.887 | 21.298 | 22.579 | 20.384 |



\end{document}

%% file: dimmatch-mgeop.tex
\subsection{Review of GEO-P}
The geometric-protean model (GEO-P) model is a model for online social
networks which incorporates geometric and ranking information into an
evolving network structure.  More specifically, the GEO-P model, as defined by Bonato,
Janssen, and Pra\l at~\cite{Bonato-2012-geop}, defines a sequence of graphs
$\set{G_t \colon t \geq 0}$ on $n$ nodes where $G_t = (V_t, E_t)$,
based on four parameters: the attachment strength $\alpha \in (0,1)$, the
density parameter $\beta \in (0, 1-\alpha)$, the dimension $m \in \N$,
and the link probability $p \in (0,1]$.   Each node $v \in V_t$ has a 
unique rank $r(v,t) \in [n]$ where $[n] =\set{1,2,\ldots, n}$; 
we explicitly list $r(v,t)$ to emphasize that the rank may change
with time.  In order
to stay consistent with the standard usage, the highest rank is $1$
and the lowest rank is $n$.  Additionally, each node has a geometric
location in $[0,1]^m$ under the torus metric $d(\cdot,\cdot)$.  That
is, for any two points $x, y \in [0,1]^m$, $d(x,y)$ is defined to be
$\min\set{\norm[\infty]{x-y -u} \colon u \in \set{-1,0,1}^m }.$  We
note that this implies that the geometric space is symmetric in
any point as the metric ``wraps'' around.  For any node $v$, we define
its influence region at time $t \geq 0$, written $R(v,t)$, to be the
ball of Euclidean volume $r(v,t)^{-\alpha}n^{-\beta}$ centered at $v$.  Notice
that, since the we are in the torus metric, this is a cube measuring
$r(v,t)^{-\nicefrac{\alpha}{m}}n^{-\nicefrac{\beta}{m}}$ on a
side. 

\begin{note}
All asymptotic results in this paper are with respect to $n$. We say that a statement holds \emph{with extremely high probability}, if it holds with probability at least \[ {1-\exp(-\omega(n) \log n)} \] for some function $\omega(n)$ with $\omega(n) \to \infty$ as $n \to \infty$. In particular, if there are a polynomial number of events, each of which holds with extremely high probability, then all of them hold with extremely high probability.
\end{note}

Let $G_0$ be any graph. In order to form $G_t$ from $G_{t-1}$, first choose a node $w$ uniformly at random
from $V_{t-1}$ and remove it. The remaining nodes are re-ranked, that is, all nodes with lower ranks than $w$ decrease their ranks by $1$. Then place a node $v$ uniformly at
random in $[0,1]^m$, generate uniformly at random a rank for $v$, and re-rank the remaining nodes again. Finally, for every node $u$ which is such that $v$ is in the
influence region of $u$, add the edge $\set{u,v}$ with probability
$p$. It is clear that this process depends only on the
current state of $G_t,$ and so forms a ergodic Markov chain with a limiting
distribution $\pi$.  A random instance of GEO-P is then defined to be
a sample from this limiting distribution.  

 It is clear that the distributions of edges of $G_t$
are determined by the relative rank histories of all the nodes at
the time the other nodes entered.  More specifically, if we order
the nodes of $G_t$ according to their age with node $1$ being the
oldest, then for any $i > j$ the probability of the edge $\set{i,j}$
being present is determined by their respective geometric locations and
the rank of node $j$ when node $i$ arrives.  Thus, in order to
sample from the limiting distribution $\pi$ it suffices to sample from
the distributions of node histories, then randomly assign
locations to the nodes, and determine if the edges are present. We
note that according to the distribution $\pi$ the final permutation between ages and ranks is 
uniformly distributed over all permutations.  Since there are $n!$
permutations of nodes and at most $n^2$ different permutations
reachable from a given state, it takes at least $\log_{n^2}(n!) =
\frac{n}{2}(1-\lilOh{1})$ iterations to reach the stationary
distribution. Standard results in the mixing rate of 
random graphs suggest that in order to assure that a sample is close
to the stationary distribution at least $\bigOmega{\log\left(
    \frac{n!}{n^2}\right)} = \bigOmega{n\log(n)}$ iterations are
required. In fact, it is easy to see that the stationary distribution is reached at the time when the last node from the initial graph $G_0$ is removed, which happens with probability $1+o(1)$ after $(1+o(1)) n \log n$ steps, by the coupon collector problem. 
 
\subsection{Introducing MGEO-P}
For large $n$ this number of iterations is a significant computational
roadblock, so we introduce here a variant of the GEO-P model which we call a
memoryless geometric-protean graph (MGEO-P).  In essence this model is
the GEO-P model where the each node has forgotten its history of ranks.
More specifically, a permutation $\sigma$ on $[n]$ is chosen uniformly
at random and $\sigma(i)$ represents the rank of the $i^{\textrm{th}}$
oldest node.  Thus, for each pair $i > j$ the edge $\set{i,j}$ is
potentially present if and only if the node $j$ is in the ball of
volume $\sigma(i)^{-\alpha}n^{-\beta}$ centered around node $i$.  It
is worth noting that, as shown in Bonato et al.~Lemma 5.2\cite{Bonato-2012-geop}, if a node in
the GEO-P model
receives an initial rank $R \geq \sqrt{n}\log^2 n$, then its rank is
$R\left(1+\bigOh{\log^{-\nicefrac{1}{2}}(n)}\right) = R (1+o(1))$ for its entire
lifetime with extremely high probability.  Thus, if we imagine
coupling the MGEO-P model in the natural way to GEO-P, and assuming that ranks do not change much as mentioned above, we have that for
all but a vanishing fraction of the edges, the probability that a
given edge is present in one model but not the other is $\bigOh{pn^{-\nicefrac{\alpha+2\beta}{2}}\log(n)^{\nicefrac{1-4\alpha}{2}}}$.
Hence, we would intuitively expect that the MGEO-P model would not
differ too much from GEO-P model.  In order to confirm this we prove
that the parameters we are interested in do not differ by much from
the proven parameters of the GEO-P model. 
Specifically, we look at the average degree, the degree distribution, and the diameter.

\subsection{An equivalent description of the MGEO-P model}
We now describe a model that is equivalent to the MGEO-P model just introduced, but that we found useful for our analysis. It has a different interpretation. The key change is that we reverse the way links are formed: when a node $i$ arrives in the network, then all existing nodes $j$ form links to $i$ if $i$ is within the influence regions of $j$. Intuitively, this models how links may arise in a citation network -- a new paper links to those that are topically related (that is, nearby in the metric space) or highly influential. In the language we used above, this process is: fix a permutation $\sigma$ on $[n]$ chosen uniformly at random and $\sigma(i)$ represents the rank of the $i^{\textrm{th}}$ oldest node.  Thus, for each pair $i > j$ the edge $\set{i,j}$ is potentially present if and only if the node $i$ is in the ball of
volume $\sigma(j)^{-\alpha}n^{-\beta}$ centered around node $j$. The two descriptions are equivalent as we can simply reverse the order of vertex arrivals. Thus, they induce the same distribution over graphs because the order is a uniform random choice.

\subsection{The average degree}

In order to consider the degrees, we first need the
following standard result on the tails of the hypergeometric
distribution, see for instance~Jansen et al.\cite{Jansen-2000-random-graphs}

\begin{lemma}\label{L:Hypergeom}
Let $X$ be the number of red balls in a set of $t$ balls chosen at
random from a set of $n$ balls containing $m$ red balls. Then, $\expect{X} = \frac{tm}{n}$, and 
for any $\epsilon > 0$, 
\[\prob{X \geq
  (1+\epsilon)\frac{tm}{n}} \leq
e^{-\frac{\epsilon^2}{2+\frac{2\epsilon}{3}}\frac{tm}{n}}.\]  Further,
for any $\epsilon \in (0,1)$, 
\[  \prob{X \leq (1-\epsilon)\frac{tm}{n}} = e^{-\frac{\epsilon^2}{2}\frac{tm}{n}}.\]
\end{lemma}

\begin{theorem}\label{T:degree}
Let $\alpha \in (0,1), \beta \in (0,1-\alpha), n \in \N, m \in \N,p \in
(0,1]$. Let $v$ be a node of $\mgeop$ with rank $R$ and age $i$, then 
\[ \deg(v) = \paren{\frac{i-1}{n-1}\frac{p}{1-\alpha}n^{1-\alpha-\beta} +
(n-i)pR^{-\alpha}n^{-\beta}}\paren{1+\bigOh{\sqrt{\frac{\log^2(n)}{n^{1-\alpha-\beta}}}}}, \]
with extremely high probability. 
\end{theorem}

\begin{proof}
Let $\deg^+(v)$ denote the number of older neighbors of $v$ and let
$\deg^-(v)$ denote the younger neighbors of $v$.  In order to
determine $\deg^+(v)$ we consider connecting $v$ to nodes of all
ranks other than $R$ and keeping $i-1$ of those uniformly at random.
The expected degree of $v$ before the edge deletion is 
\[ \sum_{r=1}^n pr^{-\alpha}n^{-\beta} - pR^{-\alpha}n^{-\beta} =
pn^{-\beta}\int_1^n x^{-\alpha}\! dx + \bigOh{1} =
\frac{p}{1-\alpha}n^{1-\alpha-\beta}+\bigOh{1}.\]
Thus, $\expect{\deg^+(v)} = \frac{i-1}{n-1}
  \frac{p}{1-\alpha}n^{1-\alpha-\beta} + \bigOh{1}$.  Furthermore, by
  Chernoff bounds the initial degree of $v$ is
  $\frac{p}{1-\alpha}n^{1-\alpha-\beta}\paren{1+\bigOh{\frac{\log(n)}{\sqrt{n^{1-\alpha-\beta}}}}}$
     with extremely high probability, and thus, Lemma~\ref{L:Hypergeom} gives 
      that \[\deg^+(v) \leq
      \frac{i-1}{n-1}\frac{p}{1-\alpha}n^{1-\alpha-\beta}\paren{1+\bigOh{\sqrt{\frac{\log^2(n)n^{\alpha+\beta}}{i}}}}\]
       with extremely high probability as well. Additionally, if $i \geq
      \log^3(n)n^{\alpha+\beta}$, then equality holds. 

Since the edge probability between $v$ and the younger nodes does
not depend on the rank of the younger neighbors, $\deg^-(v)$ can be
expressed as a sum of independent random variables which has
expectation $(n-i)pR^{-\alpha}n^{-\beta}$.  Hence, by Chernoff bounds it follows that
with extremely high probability \[\deg^-(v) \leq
(n-i)pR^{-\alpha}n^{-\beta}\paren{1 +
  \bigOh{\sqrt{\frac{\log^2(n)R^{\alpha}n^{\beta}}{n-i}}}}.\]  Now if
$n-i \geq \log^3(n)R^{\alpha}n^{\beta}$, then equality holds. 
Combining $\deg^+(v)$ and $\deg^-(v)$ we have that with extremely high probability
\[ \deg(v) \leq \frac{i-1}{n}\frac{p}{1-\alpha}n^{1-\alpha-\beta} +
(n-i)pR^{-\alpha}n^{-\beta} +
\bigOh{\sqrt{\frac{\log^2(n)i}{n^{\alpha+\beta}}} +
  \sqrt{\frac{\log^2(n)(n-i)}{R^{\alpha}n^{\beta}}}}.\] 
In order to express the error in a multiplicative faction, we note
that 
\begin{align*}
\frac{\log^2(n)i}{n^{\alpha + \beta}} \paren{\frac{n^{\alpha +
      \beta}}{i-1}}^2 &\in
\bigOh{\frac{\log^2(n)}{n^{1-\alpha-\beta}}}& i &\geq \frac{n}{2},
\\
\frac{\log^2(n)i}{n^{\alpha +
    \beta}} \paren{\frac{R^{\alpha}n^{\beta}}{n-i}}^2 &\in
\bigOh{\frac{\log^2(n)}{n^{1-\alpha-\beta}}}   & i &\leq
\frac{n}{2}, \\
\frac{\log^2(n)(n-i)}{R^{\alpha}n^{\beta}} \paren{\frac{R^{\alpha}n^{\beta}}{n-i}}^2
&\in \bigOh{\frac{\log^2(n)}{n^{1-\alpha-\beta}}}& i &\leq n
-n\paren{\frac{R}{2n}}^{\alpha}, \\
\frac{\log^2(n)(n-i)}{R^{\alpha}n^{\beta}} \paren{\frac{n^{\alpha}n^{\beta}}{i-1}}^2
&\in \bigOh{\frac{\log^2(n)}{n^{1-\alpha-\beta}}}& i &\geq n
-n\paren{\frac{R}{2n}}^{\alpha}.
\end{align*}
Thus, for the entire range of $i$ both of the error terms are
individually dominated by one of the primary terms and hence, 
we have that with extremely high probability 
\[\deg(v) \leq \paren{\frac{i-1}{n-1}\frac{p}{1-\alpha}n^{1-\alpha-\beta} +
(n-i)pR^{-\alpha}n^{-\beta}}\paren{1+\bigOh{\sqrt{\frac{\log^2(n)}{n^{1-\alpha-\beta}}}}},\]
and furthermore if $\log^3(n)n^{\alpha+\beta} \leq i \leq
n-\log^3(n)R^{\alpha}n^{\beta}$, then equality holds. 

Noting that \[\frac{\expect{\deg^+(v)}}{\expect{\deg^-(v)}} =
  \frac{n}{n-1}\frac{i-1}{n-i}
  \frac{1}{1-\alpha} \paren{\frac{R}{n}}^{\alpha} \in
  \bigOh{\frac{\log^3(n)}{n^{1-\alpha-\beta}}} \subset \lilOh{\sqrt{\frac{\log^2(n)}{n^{1-\alpha-\beta}}}}\] where $i \leq
  \log^3(n)n^{\alpha+\beta}$,  and \[ \frac{\expect{\deg^-(v)}}{\expect{\deg^+(v)}} =
  \frac{n-1}{n}\frac{n-i}{i-1}
  (1-\alpha)\paren{\frac{n}{R}}^{\alpha} \in
  \bigOh{\frac{\log^3(n)}{n^{1-\alpha-\beta}}} \subset \lilOh{\sqrt{\frac{\log^2(n)}{n^{1-\alpha-\beta}}}},\] where $n -i \geq
  \log^3(n)R^{\alpha}n^{\beta}$ completes the proof.
\end{proof}

\begin{theorem}
Let $\alpha \in (0,1), \beta \in (0,1-\alpha), n \in \N, m \in \N,$ and $p \in
(0,1]$, then with extremely high probability the average degree of node of $\mgeop$ is
\[d = \frac{p}{1-\alpha}n^{1-\alpha-\beta}\paren{1+\bigOh{\sqrt{\frac{\log^2(n)}{n^{1-\alpha-\beta}}}}}.\]
\end{theorem}

\begin{proof}
From the proof of Theorem \ref{T:degree} we have that with extremely
high probability for a node $v$
with age $i$, 
\[ \deg^+(v) \leq
\frac{i-1}{n-1}\frac{p}{1-\alpha}n^{1-\alpha-\beta}\paren{1+\bigOh{\sqrt{\frac{\log^2(n)}{n^{1-\alpha-\beta}}}}},\]
with equality if $i \geq \log^3(n)n^{\alpha+\beta}$.
Now since every edge is counted exactly once in $\deg^+(u)$ for some
node $u$, the average degree is with extremely high probability
\begin{align*}
\frac{2\size{E}}{n} &= \frac{2}{n}\sum_v \deg^+(v) \\
&\leq \frac{2}{n} \sum_{i=1}^n
\frac{i-1}{n-1}\frac{p}{1-\alpha}n^{1-\alpha-\beta}\paren{1+\bigOh{\sqrt{\frac{\log^2(n)}{n^{1-\alpha-\beta}}}}}
\\
&= \paren{1+\bigOh{\sqrt{\frac{\log^2(n)}{n^{1-\alpha-\beta}}}}}\frac{2pn^{-\alpha-\beta}}{(n-1)(1-\alpha)}
\sum_{i=1}^n (i-1) \\
&= \paren{1+\bigOh{\sqrt{\frac{\log^2(n)}{n^{1-\alpha-\beta}}}}}\frac{2pn^{-\alpha-\beta}}{(n-1)(1-\alpha)}
\binom{n}{2} \\ 
&= \frac{p}{1-\alpha}n^{1-\alpha-\beta}\paren{1+\bigOh{\sqrt{\frac{\log^2(n)}{n^{1-\alpha-\beta}}}}}.
 \\ 
\end{align*}
In a similar manner, we find that
\begin{align*}\frac{2\size{E}}{n} &\geq
\paren{1+\bigOh{\sqrt{\frac{\log^2(n)}{n^{1-\alpha-\beta}}}}}\paren{\frac{p}{1-\alpha}n^{1-\alpha-\beta}
- \frac{2pn^{-\alpha-\beta}}{(n-1)(1-\alpha)}
\binom{\log^3(n)n^{\alpha+\beta}}{2}} \\
&= \paren{1+\bigOh{\sqrt{\frac{\log^2(n)}{n^{1-\alpha-\beta}}}}}\paren{1+\bigOh{\frac{\log^6(n)}{n^{2-2\alpha-2\beta}}}}\frac{p}{1-\alpha}n^{1-\alpha-\beta}
\\
&=
\frac{p}{1-\alpha}n^{1-\alpha-\beta}\paren{1+\bigOh{\sqrt{\frac{\log^2(n)}{n^{1-\alpha-\beta}}}}},
\end{align*}
completing the proof.
\end{proof}

\subsection{The degree distribution}
Let $N_j$ be the number of nodes in $\mgeop$ with degree precisely
$j$ and let $N_{\geq k} = \sum_{j = k}^{\infty} N_j$ be the number of
nodes in degree $\mgeop$ with degree at least $k$.  We will show
that similarly to the geometric protean graphs, $N_{\geq k} \propto
k^{-\frac{1}{\alpha}}$ for a significant range of $k$, and thus,
$\mgeop$ exhibits a power-law degree distribution over that range with
power-law exponent $1+\frac{1}{\alpha}$.  Following prior work\cite{Bonato-2012-geop} we will
characterize the pairs $(i,R)$ of ages and ranks which will assure
that the degree of a node is at least $k$ and show that this value
concentrates about its expectation using the following specialization
of the Azuma-Hoeffding inequality.

\begin{theorem}[McDiarmid's Inequality]\label{T:McDiarmid}
If $X_1, X_2, \ldots, X_n$ are independent random variables and
$f(x_1, x_2, \ldots,x_n)$ is a function such that for every $i \in [n]$
\[
\abs{ f(x_1,\ldots,
  x_i, \ldots, x_n) - f(x_1, \ldots, \hat{x}_i, \ldots, x_n) }\leq
c_i,\]
then for any $\epsilon > 0$
\[ \prob{\abs{f(X_1,X_2,\ldots, X_n) - \expect{f(X_1,X_2,\ldots,X_n)}} >
  \epsilon} < 2e^{-\frac{\epsilon^2}{\sum_i c_i^2}}.\]
\end{theorem}

We will use the notation $f(n) \gg g(n)$ if $f(n)/g(n) \to \infty$ as $n \to \infty$. Similarly, $f(n) \ll g(n)$ if $g(n)/f(n) \to \infty$ as $n \to \infty$. Moreover, it will be convenient not to worry about less significant factors, so we will use $\TbigOh{f(n)}$ to denote any function which is at most $f(n) \log^{O(1)} n$.

\begin{theorem}\label{T:degree_dist}
Let $\alpha \in (0,1), \beta \in (0,1-\alpha), n \in \N, m \in \N,p \in
(0,1]$, and let $k$ and $\epsilon$ be such that 
$n^{1-\alpha-\beta} \ll k \ll
\frac{n^{1-\frac{\alpha}{2}-\beta}}{\log^{\alpha}(n)}$,
$\sqrt{\frac{\log^2(n)}{n^{1-\alpha-\beta}}}, \frac{n^{1-\alpha-\beta}}{k} \ll \epsilon$, and
$\epsilon \geq c
\log(n)k^{\frac{1}{\alpha}}n^{\frac{1}{2} - \frac{1-\beta}{\alpha}}$
for some $c > 0$, then with extremely high probability $\mgeop$
satisfies that 
\[ N_{\geq k} = \paren{1 + \bigOh{\epsilon}}\frac{\alpha}{1+\alpha}p^{\nicefrac{1}{\alpha}}n^{\nicefrac{(1-\beta)}{\alpha}}k^{-\nicefrac{1}{\alpha}}.\] 
\end{theorem}

\begin{proof}
We first note that if the age rank pair $(i,R)$ for a node $v$ satisfies that 
\[ \frac{R}{n} \leq (1-\epsilon)\paren{p
  n^{1-\alpha-\beta}\frac{n-i}{n}\frac{1}{k}}^{\frac{1}{\alpha}},\]
then by Theorem \ref{T:degree} with extremely high probability 
\begin{align*}
\deg(v) &= \paren{\frac{i-1}{n-1}\frac{p}{1-\alpha}n^{1-\alpha-\beta}
  +
  (n-i)pR^{-\alpha}n^{-\beta}}\paren{1+\bigOh{\sqrt{\frac{\log^2(n)}{n^{1-\alpha-\beta}}}}}
\\
&= \paren{\frac{i-1}{n-1}\frac{p}{1-\alpha}n^{1-\alpha-\beta}
  +
  \frac{n-i}{n}p\paren{\frac{R}{n}}^{-\alpha}n^{1-\alpha-\beta}}\paren{1+\bigOh{\sqrt{\frac{\log^2(n)}{n^{1-\alpha-\beta}}}}}
\\
&\geq \paren{\frac{i-1}{n-1}\frac{p}{1-\alpha}n^{1-\alpha-\beta}
  +
  \frac{k}{(1-\epsilon)^{\alpha}}}\paren{1+\bigOh{\sqrt{\frac{\log^2(n)}{n^{1-\alpha-\beta}}}}}
  \\
&=
\frac{k}{(1-\epsilon)^{\alpha}}\paren{1+\lilOh{\epsilon}} 
\\
&= \frac{k}{(1-\epsilon)^{\alpha}} + \lilOh{\epsilon k} \\
&= k + \alpha\epsilon k + \lilOh{\epsilon k} \\
&> k.
\end{align*}
Similarly, if 
\[ \frac{R}{n} \geq (1+\epsilon)\paren{p
  n^{1-\alpha-\beta}\frac{n-i}{n}\frac{1}{k}}^{\frac{1}{\alpha}},\]
then with extremely high probability  $\deg(v) < k$.

Let $X_i$ be the event that the node with age $i$ has rank $R$
satisfying \[R\leq (1-\epsilon)n\paren{p
  n^{1-\alpha-\beta}\frac{n-i}{n}\frac{1}{k}}^{\frac{1}{\alpha}},\]
and let $Y_i$ be the event that the node with age $i$ has rank $R$
satisfying 
\[ (1-\epsilon)n\paren{p
  n^{1-\alpha-\beta}\frac{n-i}{n}\frac{1}{k}}^{\frac{1}{\alpha}} \leq R\leq (1+\epsilon)n\paren{p
  n^{1-\alpha-\beta}\frac{n-i}{n}\frac{1}{k}}^{\frac{1}{\alpha}}.\]
Letting $X = \sum_i X_i$ and $Y = \sum_i Y_i$ we have that $X \leq
N_{\geq k} \leq X + Y$.  Thus, consider 
\begin{align*}
\expect{X} &= \sum_i \expect{X_i} \\ 
&=  \sum_i  (1-\epsilon)\paren{p
  n^{1-\alpha-\beta}\frac{n-i}{n}\frac{1}{k}}^{\frac{1}{\alpha}} \\
&= (1-\epsilon)\paren{\frac{pn^{-\alpha-\beta}}{k}}^{\frac{1}{\alpha}}
\sum_i (n-i)^{\frac{1}{\alpha}} \\
&= (1-\epsilon)\paren{\frac{pn^{-\alpha-\beta}}{k}}^{\frac{1}{\alpha}}
\paren{\frac{\alpha}{1+\alpha}n^{\frac{1+\alpha}{\alpha}} + \bigOh{1}} \\
&= (1-\epsilon)\frac{\alpha}{1+\alpha}n\paren{\frac{pn^{1-\alpha-\beta}}{k}}^{\frac{1}{\alpha}}+ \lilOh{\paren{\frac{\epsilon}{n}}^{\frac{1}{\alpha}}}.
\end{align*} 
We note as well that $\expect{Y} =
\frac{2\epsilon}{1-\epsilon}\expect{X}$.

We recall that the age-rank pairs can be represented by a permutation $\sigma$
chosen uniformly at random from the symmetric group, and thus, it can
be generated by a sequence of transpositions $(1,a_1)(2,a_2)\cdot
(n,a_n)$ where each $a_i$ is chosen independently and uniformly at
random from $\set{i,i+1, \ldots, n}$.  Thus, $X$ (and $Y$) may be
viewed as a function of independent random variables and so Theorem~\ref{T:McDiarmid} applies.  Furthermore, the change of any
particular variable impacts the value of $X$ by at most 2.  Hence, we note that,
\[ \frac{1}{n} \epsilon^2\expect{X}^2 \geq \epsilon^2 (1-\epsilon)^2\frac{\alpha^2}{(1+\alpha)^2}
n \paren{\frac{pn^{1-\alpha-\beta}}{k}}^{\frac{2}{\alpha}} \in
\bigOmega{\log^2(n)},\] and thus, with extremely high probability
$\abs{X - \expect{X}} \leq \epsilon \expect{X}$ and $\abs{Y -
  \expect{Y}} \leq \epsilon \expect{X}$.  Hence, we have that 
\[ \expect{X} - \epsilon\expect{X} \leq N_{\geq k} \leq \expect{X} +
3\epsilon\expect{X},\] with extremely high probability and the desired
result follows.
\end{proof}

We note that by choosing $\epsilon = \log^{-\nicefrac{1}{3}}(n)$ we
can easily obtain the same type of degree distribution result for MGEO-P that exists for the original GEO-P\cite{Bonato-2012-geop}. 
\begin{theorem}
Let $\alpha \in (0,1), \beta \in (0,1-\alpha), n \in \N, m \in \N,p \in
(0,1]$, and 
\[ n^{1-\alpha-\beta}\log^{\nicefrac{1}{2}}(n) \leq k \leq
n^{1-\nicefrac{\alpha}{2}-\beta}\log^{-2\alpha - 1}(n),\] then with
extremely high probability $\mgeop$ satisfies \[ N_{\geq k} = \left( 1
  + \bigOh{\log^{-\nicefrac{1}{3}}(n)}\right)
\frac{\alpha}{1+\alpha}p^{\nicefrac{1}{\alpha}}n^{\nicefrac{(1-\beta)}{\alpha}}k^{-\frac{1}{\alpha}},\]
where $N_{\geq k}$ is the number of nodes of degree at least $k$.
\end{theorem}

\subsection{The diameter}

\begin{theorem}
Let $\alpha \in (0,1), \beta \in (0,1-\alpha), n \in \N, m \in \N,p \in
(0,1]$.  The diameter of $\mgeop$ is $n^{\bigTheta{\frac{1}{m}}}$ with extremely high probability.
\end{theorem}
\begin{proof}
We first show that the diameter is
$\TbigOh{n^{\frac{\beta}{(1-\alpha)m}}} \in n^{\bigOh{\frac{1}{m}}}$.
  To this end, let 
  $$
  t = 10 \floor{ \left( n^{\frac{\beta}{(1-\alpha)}} \ln^{\frac{2\alpha}{1-\alpha}}(n) \right)^{1/m}}
  $$ and
  divide $[0,1)^m$ into $t^m$ uniform subcubes with side-lengths
  $\frac{1}{t}$ in the natural way.  Now, as $\frac{\beta}{(1-\alpha)}
  < 1$, by Chernoff bounds there are
  $\TbigTheta{n^{1-\frac{\beta}{(1-\alpha)}}}$ nodes in each of the
    subcubes with extremely high probability.  Thus, in order to show
    diameter $\TbigOh{n^{\frac{\beta}{(1-\alpha)m}}}$ it suffices to
    show that for any two nodes $u$ and $v$ at $\ell_{\infty}$-distance at most $\nicefrac{2}{t}$ the graph distance between the
    two nodes is at most some fixed constant.

Now consider an arbitrary node $v$.   By Chernoff bounds, with
extremely high probability there are
 $\bigOmega{n^{1-\alpha-\beta}}$ nodes at $\ell_{\infty}$
distance at most $\frac {1}{2} n^{-\frac{\alpha+\beta}{m}}$ from $v$ 
and age rank at least $\frac{n}{2}$ (that is, young nodes).
As the radius of influence of every node is at least
$\frac{1}{2}n^{-\frac{\alpha+\beta}{m}}$, this implies that with
extremely high probability every node has 
$\bigOmega{n^{1-\alpha-\beta}}$ neighbors at $\ell_{\infty}$-distance
at most $\frac 12 n^{-\frac{\alpha+\beta}{m}}$ with age rank at
least $\frac{n}{2}$.   

In a similar manner, by combining Lemma~\ref{L:Hypergeom} and Chernoff bounds,
with extremely high probability every node $v$ with age rank at least
$\frac{n}{2}$ has 
$$
\bigOmega{ (1/t)^m n^{\beta/(1-\alpha)} \ln^{2/(1-\alpha)} n} = \bigOmega{ \ln^{-2\alpha/(1-\alpha)+2/(1-\alpha)} n} = \bigOmega{ \ln^2 n}
$$ 
neighbors at
$\ell_{\infty}$-distance at most
$\frac {1}{2t}$, with rank at most
$n^{\frac{\beta}{1-\alpha}}  \ln^{2/(1-\alpha)} n$, and age rank at most $n/2$ (that is, old nodes).  
Note that each node with rank at most $n^{\frac{\beta}{1-\alpha}}  \ln^{2/(1-\alpha)} n$ has radius of influence at least 
$$
\frac {1}{2} \left( n^{\frac{-\alpha \beta}{1-\alpha}-\beta}  \ln^{-2\alpha/(1-\alpha)} n \right)^{1/m} = (1+o(1)) \frac {5}{t}.
$$

Combining these two observations we have that, with
extremely high probability, every node $v$ is within graph-distance two 
and $\ell_{\infty}$-distance $n^{-\frac{\alpha+\beta}{m}} +
\frac {1}{2t}$ of a set of 
$\bigTheta{\ln^2(n)}$ nodes, $X_v$, with rank at most
$n^{\frac{\beta}{1-\alpha}}  \ln^{2/(1-\alpha)} n$.  Thus, if $u$
and $v$ are at $\ell_{\infty}$-distance at most $\frac{2}{t}$, the
distance between elements of $X_u$ and $X_v$ is at most
$\frac{2}{t} + 2n^{-\frac{\alpha+\beta}{m}} +
\frac{1}{t} \le \frac {4}{t}$.  
On the other hand, as we already mentioned, the radius of influence of each node in $X_u$ or $X_v$ is at least $4/t$.
Thus, with extremely high
  probability, some member of $X_u$ and $X_v$ are adjacent and hence
  $u$ and $v$ are within graph distance $5$, completing the proof of
  the upper bound.

For the lower bound, let us take some node $v$ and consider distances to other nodes. With probability $1-2^{-m}$ some other node is at $\ell_{\infty}$-distance at least 1/4. Hence, by Chernoff bounds, with extremely high
probability there exist two nodes at $\ell_{\infty}$-distance at
least $\frac{1}{4}$.  As the diameter of every influence region is at
most $n^{-\frac{\beta}{m}}$, this gives that the diameter of the graph 
is $\bigOmega{n^{\frac{\beta}{m}}} \in n^{\Omega(1/m)}$.
\end{proof}

%% file: dimmatch-arxiv.bbl
\begin{thebibliography}{10}

\bibitem{Banerjee-2012-distance}
Anirban Banerjee.
\newblock Structural distance and evolutionary relationship of networks.
\newblock {\em Biosystems}, 107(3):186 -- 196, March 2012.

\bibitem{Banerjee-2009-graph-spectra}
Anirban Banerjee and J\"{u}rgen Jost.
\newblock Graph spectra as a systematic tool in computational biology.
\newblock {\em Discrete Applied Mathematics}, 157(10):2425 -- 2431, 2009.
\newblock Networks in Computational Biology.

\bibitem{barabasi1999-scaling}
Albert-L{\'a}szl{\'o} Barab{\'a}si and R{\'e}ka Albert.
\newblock Emergence of scaling in random networks.
\newblock {\em Science}, 286(5439):509--512, October 1999.

\bibitem{Bayati}
Mohsen Bayati, Jeong Kim, and Amin Saberi.
\newblock A sequential algorithm for generating random graphs.
\newblock {\em Algorithmica}, 58(4):860--910, 2010.
\newblock 10.1007/s00453-009-9340-1.

\bibitem{Bonato-2012-geop}
Anthony Bonato, Jeannette Janssen, and Pawe\l Pra{\l}at.
\newblock Geometric protean graphs.
\newblock {\em Internet Mathematics}, 8(1-2):2--28, 2012.

\bibitem{Chung-1992-book}
Fan R.~L. Chung.
\newblock {\em Spectral Graph Theory}.
\newblock American Mathematical Society, 1992.

\bibitem{clauset2009-powerlaw}
Aaron Clauset, Cosma~Rohilla Shalizi, and M.~E.~J. Newman.
\newblock Power-law distributions in empirical data.
\newblock {\em SIAM Review}, 51(4):661--703, 2009.

\bibitem{Csardi-2006-igraph}
Gabor Csardi and Tamas Nepusz.
\newblock The igraph software package for complex network research.
\newblock {\em InterJournal}, Complex Systems:1695, 2006.
\newblock Version 0.6.

\bibitem{Dhillon-1997-mrrr}
Inderjit~S. Dhillon.
\newblock {\em A new $O(n^2)$ algorithm for the symmetric tridiagonal
  eigenvalue/eigenvector problem}.
\newblock PhD thesis, University of California, Berkeley, 1997.

\bibitem{Dhillon-2006-mrrr}
Inderjit~S. Dhillon, Beresford~N. Parlett, and Christof V\"{o}mel.
\newblock The design and implementation of the mrrr algorithm.
\newblock {\em ACM Trans. Math. Softw.}, 32(4):533--560, December 2006.

\bibitem{Estrada-2006-expansion}
E.~Estrada.
\newblock Spectral scaling and good expansion properties in complex networks.
\newblock {\em EPL (Europhysics Letters)}, 73(4):649, 2006.

\bibitem{Faloutsos-1999-power-law}
Michalis Faloutsos, Petros Faloutsos, and Christos Faloutsos.
\newblock On power-law relationships of the internet topology.
\newblock {\em SIGCOMM Comput. Commun. Rev.}, 29:251--262, August 1999.

\bibitem{Gleich}
David~F. Gleich.
\newblock \url{https://www.github.com/dgleich/bisquik}.

\bibitem{Gleich-2012-kronecker}
David~F. Gleich and Art~B. Owen.
\newblock Moment based estimation of stochastic {Kronecker} graph parameters.
\newblock {\em Internet Mathematics}, 8(3):232--256, August 2012.

\bibitem{Jansen-2000-random-graphs}
Svante Janson, Thomasz {\L}uczak, and Andrzej Ruci\'{n}ski.
\newblock {\em RANDOM GRAPHS}.
\newblock John Wiley \& Sons, Inc., 2000.

\bibitem{Kim-2012-mag}
Myunghwan Kim and Jure Leskovec.
\newblock Multiplicative attribute graph model of real-world networks.
\newblock {\em Internet Mathematics}, 8(1-2):113--160, 2012.

\bibitem{Kolda-2013-BTER}
Tamara~G. Kolda, Ali Pinar, Todd Plantenga, and C.~Seshadhri.
\newblock A scalable generative graph model with community structure.
\newblock {\em arXiv}, cs.SI:1302.6636, 2013.

\bibitem{Krioukov-2012-cosmology}
Dmitri Krioukov, Maksim Kitsak, Robert~S. Sinkovits, David Rideout, David
  Meyer, and Mari\'an Bogu\~n\'a.
\newblock Network cosmology.
\newblock {\em Sci. Rep.}, 2:2012/11/16/online, 2012.

\bibitem{Krioukov-2013-growing}
Dmitri Krioukov and Massimo Ostilli.
\newblock Duality between equilibrium and growing networks.
\newblock {\em Phys. Rev. E}, 88:022808, Aug 2013.

\bibitem{Krioukov-2010-hyperbolic}
Dmitri Krioukov, Fragkiskos Papadopoulos, Maksim Kitsak, Amin Vahdat, and
  Mari\'an Bogu\~n\'a.
\newblock Hyperbolic geometry of complex networks.
\newblock {\em Phys. Rev. E}, 82:036106, Sep 2010.

\bibitem{Kumar-2000-copying}
R.~Kumar, P.~Raghavan, S.~Rajagopalan, D.~Sivakumar, A.~Tomkins, and E.~Upfal.
\newblock Stochastic models for the web graph.
\newblock In {\em Proceedings of the 41st Annual Symposium on Foundations of
  Computer Science}, FOCS '00, pages 57--65, Washington, DC, USA, 2000. IEEE
  Computer Society.

\bibitem{Leskovec-2010-KronFit}
Jure Leskovec, Deepayan Chakrabarti, Jon Kleinberg, Christos Faloutsos, and
  Zoubin Ghahramani.
\newblock Kronecker graphs: An approach to modeling networks.
\newblock {\em Journal of Machine Learning Research}, 11:985--1042, February
  2010.

\bibitem{Leskovec-2007-densification}
Jure Leskovec, Jon Kleinberg, and Christos Faloutsos.
\newblock Graph evolution: Densification and shrinking diameters.
\newblock {\em ACM Trans. Knowl. Discov. Data}, 1:1--41, March 2007.

\bibitem{Leskovec-2009-community-structure}
Jure Leskovec, Kevin~J. Lang, Anirban Dasgupta, and Michael~W. Mahoney.
\newblock Community structure in large networks: Natural cluster sizes and the
  absence of large well-defined clusters.
\newblock {\em Internet Mathematics}, 6(1):29--123, September 2009.

\bibitem{McPherson-1991-Blau}
J.~Miller McPherson and James~R. Ranger-Moore.
\newblock Evolution on a dancing landscape: Organizations and networks in
  dynamic blau space.
\newblock {\em Social Forces}, 70(1):19--42, 1991.

\bibitem{McPherson-2001-homophily}
Miller McPherson, Lynn Smith-Lovin, and James~M. Cook.
\newblock Birds of a feather: Homophily in social networks.
\newblock {\em Annual Review of Sociology}, 27:415--444, 2001.

\bibitem{Memisevic-2010-modeling}
Vesna Memi\v{s}evi\'{c}, Tijana Milenkovi\'{c}, and Nata\v{s}a Pr\v{z}ulj.
\newblock An integrative approach to modeling biological networks.
\newblock {\em Journal of Integrative Bioinformatics}, 7(3):120, 2010.

\bibitem{Moreno-2013-mom}
Sebastian~I. Moreno, Jennifer Neville, and Sergey Kirshner.
\newblock Learning mixed kronecker product graph models with simulated method
  of moments.
\newblock In {\em Proceedings of the 19th ACM SIGKDD International Conference
  on Knowledge Discovery and Data Mining}, KDD '13, pages 1052--1060, New York,
  NY, USA, 2013. ACM.

\bibitem{Nepusz-plfit}
Tam\'{a}s Nepusz.
\newblock plfit software.
\newblock https://github.com/ntamas/plfit, 2012.

\bibitem{Palmer-2002-fast-anf}
Christopher~R. Palmer, Phillip~B. Gibbons, and Christos Faloutsos.
\newblock Anf: a fast and scalable tool for data mining in massive graphs.
\newblock In {\em KDD '02: Proceedings of the eighth ACM SIGKDD international
  conference on Knowledge discovery and data mining}, pages 81--90, New York,
  NY, USA, 2002. ACM.

\bibitem{Patro-2012-ghost}
Rob Patro and Carl Kingsford.
\newblock Global network alignment using multiscale spectral signatures.
\newblock {\em Bioinformatics}, 28(23):3105--3114, 2012.

\bibitem{Seidman1983-cores}
Stephen~B. Seidman.
\newblock Network structure and minimum degree.
\newblock {\em Social Networks}, 5(3):269--287, 1983.

\bibitem{Sweeney-2000-uniqueness}
L.~Sweeney.
\newblock Uniqueness of simple demographics in the u.s. population.
\newblock Technical Report LIDAPWP4, Carnegie Mellon University, 2000.

\bibitem{Traud-2011-facebook}
Amanda~L. Traud, Peter~J. Mucha, and Mason~A. Porter.
\newblock Social structure of facebook networks.
\newblock {\em arXiv}, cs.SI:1102.2166, February 2011.

\bibitem{Vomel-2010-MRRR}
Christof V\"{o}mel.
\newblock Scalapack's mrrr algorithm.
\newblock {\em ACM Trans. Math. Softw.}, 37(1):1:1--1:35, January 2010.

\bibitem{watts1998-dynamics}
Duncan~J. Watts and Steven~H. Strogatz.
\newblock Collective dynamics of ``small-world'' networks.
\newblock {\em Nature}, 393(6684):440--442, June 1998.

\bibitem{Wernicke-2006-motifs}
Sebastian Wernicke.
\newblock Efficient detection of network motifs.
\newblock {\em Computational Biology and Bioinformatics, IEEE/ACM Transactions
  on}, 3(4):347--359, 2006.

\bibitem{witten2005-weka}
Ian~H. Witten and Eibe Frank.
\newblock {\em Data Mining: Practical machine learning tools and techniques}.
\newblock Morgan Kaufmann, 2005.

\bibitem{Zhang-2006-Apollonian}
Zhongzhi Zhang, Francesc Comellas, Guillaume Fertin, and Lili Rong.
\newblock High-dimensional apollonian networks.
\newblock {\em Journal of Physics A: Mathematical and General}, 39(8):1811,
  2006.

\bibitem{Zhao-2011-rigel}
Xiaohan Zhao, A.~Sala, Haitao Zheng, and B.Y. Zhao.
\newblock Efficient shortest paths on massive social graphs.
\newblock In {\em Collaborative Computing: Networking, Applications and
  Worksharing (CollaborateCom), 2011 7th International Conference on}, pages
  77--86, 2011.

\end{thebibliography}
